\documentclass[11pt,letterpaper]{article}

\title{\bf Lifting General Relativity \\ to Observer Space}
\author{\bf Steffen Gielen\\[.5em]
{\sl \small Perimeter Institute for Theoretical Physics} \\[-.3em]
{\sl \small 31 Caroline St. N.} \\[-.3em]
{\sl \small  Waterloo ON, N2L 2Y5, Canada} \\
\small  \texttt{sgielen@perimeterinstitute.ca} 
 \and
{\bf Derek K.\ \!Wise} \\[.5em]
{\sl \small Institute for Quantum Gravity} \\[-.3em]
{\sl \small Universit\"at Erlangen--N\"urnberg} \\[-.3em]
{\sl \small Staudtstr.\ \!7/B2,\! 91058 Erlangen,\! Germany} \\
\small \texttt{derek.wise@gravity.fau.de} 
}

\usepackage{amssymb,amsmath,amsthm}

	\makeatletter 
	\def\th@plain{%
	  \thm@notefont{}
	  \itshape 
	}
	\def\th@definition{%
	  \thm@notefont{}
	  \normalfont 
	}
\makeatother

\usepackage[
    colorlinks,%
    linkcolor=blue,citecolor=red,urlcolor=blue,
]{hyperref}
\usepackage[all,dvips,knot,color]{xy}
\usepackage{tikz}
\usepackage{graphicx}
\usepackage{wrapfig}
\usepackage{calc}

\def\barr{\begin{array}}
\def\earr{\end{array}}
\def\ben{\begin{equation}}
\def\een{\end{equation}}
\def\bena{\begin{eqnarray}}
\def\eena{\end{eqnarray}}
 
\newcommand{\sect}[1]{\setcounter{equation}{0}\section{#1}}

\newcommand{\om}{\omega}
\newcommand{\Om}{\Omega}

\newcommand{\R}{{\mathbb R}}

\newcommand{\Z}{{\mathbb Z}}

\newcommand{\M}{{M}} 
\renewcommand{\S}{\mathcal{S}} 

\newcommand{\fb}{FM} 

\newcommand{\fake}{\mathcal{T}} 
\newcommand{\ff}{{P}} 
\newcommand{\fo}{\mathcal{O}} 

\newcommand{\GH}{Z} 

\newcommand{\gh}{z} 
\newcommand{\hk}{y} 

\newcommand{\ggh}{\mathfrak{z}} 
\newcommand{\ghk}{\mathfrak{y}} 

\newcommand{\ughk}{\underline{\ghk}} 

\newcommand{\uh}{\underline{\h}}
\newcommand{\uk}{\underline{\k}}

\newcommand{\Hyp}{\mathrm{H}} 

\newcommand{\uA}{\underline{A}}

\newcommand{\maps}{\colon}

\def\stackto #1 { \, {\stackrel{#1}{\longrightarrow}}\, }
\def\stackTo #1 { {\stackrel{#1}{\Longrightarrow}} }

\newcommand{\tr}{{\rm tr}}

\newcommand{\Ad}{{\rm Ad}}

\newcommand{\Ktild}{{\widetilde K}}
\newcommand{\Htild}{{\widetilde H}}
\newcommand{\Gtild}{{\widetilde G}}

\newcommand{\SO}{{\rm SO}}

\newcommand{\ISO}{{\rm ISO}}
\newcommand{\Iso}{\mathfrak{iso}}

\newcommand{\SIM}{{\rm SIM}}

\newcommand{\g}{\mathfrak{g}}
\newcommand{\h}{\mathfrak{h}}

\renewcommand{\k}{\mathfrak{k}}

\renewcommand{\j}{\mathfrak{j}}

\newcommand{\Spin}{{\rm Spin}}
\newcommand{\ISpin}{{\rm ISpin}}

\newcommand{\Diff}{{\rm Diff}}
\newcommand{\Vect}{{\rm Vect}}

\newcommand{\define}[1]{{\bf #1}}

\newcommand{\we}{\wedge}
\renewcommand{\L}{\pounds} 

\newcommand{\tensor}{\otimes}

\newcommand{\half}{\frac{1}{2}}

\newtheorem{thm}{Theorem}
\newtheorem{lemma}[thm]{Lemma}
\newtheorem{prop}[thm]{Proposition}
\newtheorem{defn}[thm]{Definition}

\newtheorem{exple}[thm]{Example}

\renewenvironment{proof}{\noindent
\textbf{Proof: }}{\hfill\rule{.6em}{.8em} \medskip}
\newenvironment{proof.within.proof}
{\noindent{\it Proof:}}{
\hfill $\Box$ \medskip}

\newcounter{Ccounter} 
\newenvironment{C-list}{  
\begin{list}{{\rm C\arabic{Ccounter}}.}{\usecounter{Ccounter}}
}{\end{list}}

\newcounter{Cpcounter} 
\newenvironment{C'-list}{  
\begin{list}{{\rm C\arabic{Cpcounter}${}'$}.}{\usecounter{Cpcounter}}
}{\end{list}}

\newcommand{\hepth}[1]{\href{http://arxiv.org/abs/hep-th/#1}{arXiv:hep-th/#1}}

\newcommand{\grqc}[1]{\href{http://arxiv.org/abs/gr-qc/#1}{arXiv:gr-qc/#1}}

\newcommand{\arxiv}[1]{\href{http://arxiv.org/abs/#1/}{arXiv:#1}}

\newcommand{\webpage}[1]{{\color{blue}}\url{#1}{\color{blue}}}

\begin{document}

\date{May 4, 2013}

\maketitle

\thispagestyle{empty}

\begin{abstract}

The `observer space' of a Lorentzian spacetime is the space of future-timelike unit tangent vectors.  Using Cartan geometry, we first study the structure a given spacetime induces on its observer space, then use this to define abstract observer space geometries for which no underlying spacetime is assumed.   We propose taking observer space as fundamental in general relativity, and prove integrability conditions under which spacetime can be reconstructed as a quotient of observer space. 
Additional field equations on observer space then descend to Einstein's equations on the reconstructed spacetime.  We also consider the case where no such reconstruction is possible, and spacetime becomes an observer-dependent, relative concept. Finally, we discuss applications of observer space, including a
geometric link between covariant and canonical approaches to gravity.

\end{abstract}

\section{Introduction}
\label{intro}

General relativity is about understanding that physics does not take place against the backdrop of a fixed geometry.  Rather, geometry itself is a dynamical entity, bending and curving in response to matter, just as matter is subject to geometric rules of the space it inhabits.  
There are, however, different possible interpretations of such statements.  In particular, do we mean the geometry of {\em spacetime}, the geometry of {\em space}, or something else?  

This question is the root of tension between `covariant' and `canonical' approaches to gravity. The `covariant' approach focuses on the geometry of {\em spacetime}, given `all at once'.  This is elegant, but unfortunately rather far removed from our actual experience of the world, in which space and time appear quite distinct.  The so-called `canonical' picture focuses instead on the geometry of {\em space} and how this geometry evolves in time, and is thus more clearly related to our spatiotemporal intuition.   On the other hand, the notion of `time' is fixed arbitrarily from the outset, going against the spirit of relativity, even when the final result is independent of this choice.   Worse yet, showing this independence in some formulations is decidedly nontrivial.  The term `canonical gravity', stemming from the `canonically conjugate' variables in Hamiltonian mechanics, thus stands in ironic contrast with standard mathematical use of the word `canonical', where it means involving no arbitrary choices.  In brief, canonical gravity is not canonical.  

In this paper, we reformulate general relativity in a way that maintains the best of both approaches.  To do this, we pass from spacetime to {\em observer space}---a 7-dimensional manifold of all possible observers.  On one hand, this perspective offers a clear-cut distinction between spatial and temporal directions. On the other, it acknowledges the local, observer-dependent nature of time and space in general relativity.   Because we consider the space of all possible observers, there are no arbitrary choices to be made.  

But how do we describe the `geometry of observer space', and how is it related to the geometry of spacetime or of space?  To make this precise, it is helpful to use the approach to geometry rooted in the works of Felix Klein and \'Elie Cartan.  

Klein's Erlangen Program was about understanding geometry in terms of {\em symmetry}.  In its original form, it applied to the `Platonic ideals' of geometry---homogeneous spaces, such as Euclidean, hyperbolic, or projective geometry, in which any two points look essentially the same.  Cartan generalized Klein's ideas to the setting of {\em differential} geometry, giving a precise characterization of spaces with only `infinitesimal' symmetry.  His approach involves `infinitesimally modeling' a general manifold on one of Klein's homogeneous geometries. 

This outlook meshes nicely with relativistic physics: because geometry is {\em locally}, dynamically determined, our Platonic ideals of classical geometry carry over to the real world only as  infinitesimal approximations.  In passing from special relativity to general relativity, for example, Minkowski spacetime $\R^{3,1}$ survives only as the tangent space to a more general Lorentzian manifold $\M$.   The global Lorentz symmetry of $\R^{3,1}$ itself becomes a local gauge symmetry of orthonormal frames, which essentially implements Einstein's equivalence principle.  A connection describes how the frames at different points are related, in a way that generally depends on the path between points.  So, in some sense, general relativity describes spacetimes that are {\em infinitesimally modeled on Minkowski space}.  

To make Cartan's idea of infinitesimal modeling more precise, recall that a \define{homogeneous space}, or \define{Klein geometry}, is a manifold $\GH$ equipped with a smooth transitive action of a Lie group $G$.  If $H\subseteq G$ is the subgroup fixing some point $\gh\in\GH$, then $\GH$ may be identified with the coset space $G/H$.  On the other hand, a \define{Cartan geometry} on $M$, modeled on $\GH$, locally amounts to a $\g$-valued 1-form which, when composed with the projection $\g\to \g/\h$, gives an identification of the tangent space $T_xM$ with $\g/\h \cong T_\gh \GH$.  This 1-form transforms under gauge transformations as a connection, but the geometry is only invariant under the subgroup $H$.  There is thus an intrinsic `symmetry breaking' aspect to any Cartan geometry. 

Hints of Cartan geometry are to be found in gauge theoretic formulations of general relativity.  Kibble \cite{kibble} may have been first to notice that the usual Lorentz connection and coframe field, or vierbein, can be viewed as pieces of a single connection for the Poincar\'e group $\ISO(3,1)$.  Later, MacDowell and Mansouri \cite{macdo} discovered a gravity action that uses the same idea with a connection for the de Sitter group $\SO(4,1)$ or anti-de Sitter group $\SO(3,2)$.  In each of these formulations, the naive `gauge group' $G$ of the theory is broken to the Lorentz group $H=\SO(3,1)$---a telltale signal that Cartan geometry secretly underlies these actions.  The relevant model Klein geometry $G/H$ has:
\[
  G = \left\{ 
  \begin{array}{rl}
  \SO(4,1) & \text{\sf de Sitter group} \\
  \ISO(3,1) & \text{\sf Poincar\'e group} \\
  \SO(3,2) & \text{\sf anti-de Sitter group}
  \end{array}
  \right. \qquad \quad
  H = \SO(3,1) \quad \text{\sf Lorentz group.} 
\]
MacDowell--Mansouri gravity and Poincar\'e gauge theory thus provide Cartan-geometric descriptions of the `covariant' picture of general relativity. \cite{derekmacd,dereksigma}

Similarly, on the `canonical' side,  Ashtekar variables---especially in their `real' form originally given by Barbero \cite{barbero}---also hint at Cartan geometry.   The key variables are an $\SO(3)$ connection on {\em space} together with a spatial coframe field.  This makes it tempting to think of these fields as coming from a Cartan connection with model geometry $H'/K$ where
\[
  H' \cong \left\{ 
  \begin{array}{rl}
  \SO(4) & \text{\sf spherical group} \\
  \ISO(3) & \text{\sf Euclidean group} \\
  \SO(3,1) & \text{\sf hyperbolic group}
  \end{array}
  \right. \qquad \quad
  K = \SO(3) \quad \text{\sf rotation group,} 
\]
depending on whether space is modeled on the 3-sphere, Euclidean space or hyperbolic space.  Presumably, the choice of $H'$ should be related to the choice of $G$, since the model spacetimes have different associated spatial geometries. 

However, while tempting, the precise relationship between real Ashtekar variables and Cartan geometry modeled on $H'/K$ is not immediately apparent.   In an effort to sort out this relationship, we recently obtained a version of Ashtekar variables using `spontaneous breaking' of Lorentz symmetry \cite{lorentz}.   The idea is to introduce a field of `local observers'---each with their own preferred local notions of space and time---to extract from the spacetime Lorentz connection and coframe field a `spatial' $\SO(3)$ connection and triad.  These pieces can be assembled into a `spatial Cartan connection', giving a system  of evolving spatial Cartan geometries, or `Cartan geometrodynamics'.  But the role of $H'$ is still unclear: breaking Lorentz symmetry has more to do with the coset space $H/K$ than with the spatial geometry $H'/K$.  

This becomes clear when we consider {\em observer space}.   The key geometric idea is to combine the two levels of `symmetry breaking' we have just described:  picking not just a point but a particular {\em observer} in homogeneous spacetime breaks symmetry not just to $H$ or $H'$ but all the way to $K$ in a single step.  Hence $G/K$ is the observer space of the model spacetime $G/H$.   $H$ and $H'$ still play geometric roles in observer space geometry, as do each of the possible coset spaces:
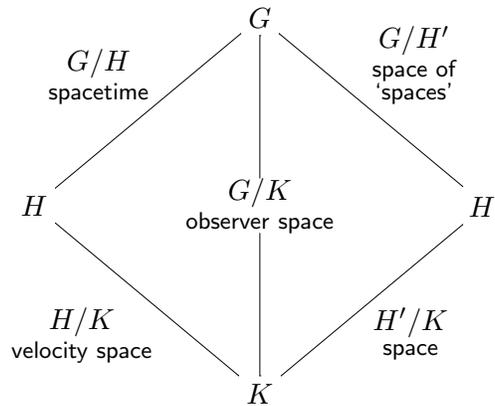
\newsavebox{\events}
\savebox{\events}{\xy (0,0)*{G/H};
    (0,-4)*{\text{\sf\footnotesize spacetime}};
    \endxy}
\newsavebox{\velocity}
\savebox{\velocity}{\xy (0,0)*{H/K};
    (0,-4)*{\text{\sf\footnotesize velocity space}};
    \endxy}
\newsavebox{\spc}
\savebox{\spc}{\xy (0,0)*{H'/K};
    (0,-4)*{\text{\sf\footnotesize space}};
    \endxy}
\newsavebox{\frames}
\savebox{\frames}{\xy (0,0)*{G/H'};
    (0,-4)*{\txt{\sf\footnotesize space of}};
    (0,-7)*{\txt{\sf\footnotesize `spaces'}};
    \endxy}
\newsavebox{\observers}
\savebox{\observers}{\xy (0,0)*{G/K};
    (0,-4)*{\text{\sf\footnotesize observer space}};
    \endxy}
\[
\xygraph{
  []!{0;<3cm,-2.5cm>:<3cm, 2.5cm>::}
  []{G}="A" :@{-}^{\usebox{\frames}} [r] 
    {H'} :@{-}^{\usebox{\spc}} [d]
    {K}="B" :@{-}^{\usebox{\velocity}} [l]
    {H} :@{-}^{\usebox{\events}} "A"
    :@{-}|-{\usebox{\observers}} "B"
  }
\]
In particular, to get from $G$ to $K$ we can just as well go through $H'$, first choosing a homogeneous `spatial slice' and then a point in this submanifold breaking the symmetry group down to $K$.   All of this generalizes from homogeneous to Cartan geometry, where we can view general observer spaces as a deformation of the homogeneous models. 

One point of this paper is that spacetime Cartan geometry and Cartan geometrodynamics are just two aspects of the {\em Cartan geometry of observer space}.  In more physical language, the geometry of observer space links the covariant and canonical pictures of general relativity.  Our study of this idea began with understanding how Ashtekar variables arise from breaking Lorentz symmetry using local observers \cite{lorentz, essay}.  The present paper is, in part, a continuation of this story, completing and clarifying the geometric picture underlying our previous work.
  
However, we also have independent physical motives for studying observer space, especially for potential applications beyond general relativity. First, since the group stabilizing a given observer is $K$, we can use observer space to investigate Lorentz-violating physical theories which take $K$ as the fundamental symmetry group. Such a violation of Lorentz symmetry is a possibility that continues to be investigated experimentally, and is also inherent in several theoretical models. For instance, several proposals for gravitational theories beyond general relativity, such as Ho\v{r}ava-Lifshitz gravity \cite{horava}, causal dynamical triangulations \cite{cdt}, and shape dynamics \cite{shape}, involve a preferred foliation of spacetime. In our framework, a choice of  foliation corresponds to a field of local observers, locally breaking Lorentz symmetry; on observer space we consider all such observers at once. More practically, spacetime geometry can only be probed by `observers', and to build a phenomenological model to be confronted with observation, one might prefer to assume gauge invariance only under $K$. While many such proposals for Lorentz violation are motivated by attempts to quantize gravity, our framework is, at this stage, purely classical.

More interestingly, passing from spacetime to observer space as the arena for physics could allow us to discard spacetime as a fundamental concept. This is the main message of the `relative locality' proposal \cite{relative}: the notion of spacetime itself may be {\em observer-dependent}. The proposal of \cite{relative} is essentially a modification of special relativity based on such an observer-dependent spacetime, with an absolute momentum space common to all observers. In a more general setting one would expect both momentum space and spacetime to be observer-dependent concepts; this will be precisely the interpretation given to a general observer space geometry in this paper. Our construction hence provides a natural framework to move from `special' to `general' relative locality. 

\subsection*{Plan of the paper}
In section \ref{sec:spacetime}, we explain how Cartan geometry is used in ordinary spacetime physics,  including a review of the Cartan-geometric underpinnings of MacDowell--Mansouri gravity. 

In section \ref{sec:observerspace}, we study the geometry of observer space, first with a Lorentzian spacetime given at the outset, and then from the perspective of Cartan geometry, where the Cartan connection induces the geometry of observer space. We explain our recent `Cartan geometrodynamics' picture of Ashtekar variables from the perspective of observer space.

Section \ref{sec:lifting-gr} contains some of our main results.  We consider the possibility that observer space is more fundamental than spacetime, and derive conditions (Thm.~\ref{thm:ex-spacet}) for the {\em existence of spacetime}, i.e.\ the ability to reconstruct an observer-independent spacetime from observer space. We give an action on observer space that allows the reconstruction of spacetime, and whose solutions include all solutions of vacuum general relativity.

In section \ref{sec:beyond} we discuss the general scenario where no such reconstruction is possible and spacetime is {\em relative}. We explain observer-dependent notions of {\em coincidence} of observers and of spacetime. As a special case we consider the possibility that velocity space, rather than spacetime, is absolute. We explain how the proposal of relative locality \cite{relative} is naturally described in these terms.

We conclude with some remarks about the fundamental status of spacetime, and suggest some directions for further investigation. For the convenience of the reader, we summarize our notation in appendix \ref{notation}.

\section{Cartan geometry in spacetime physics}
\label{sec:spacetime}

We begin with a brief overview of Cartan geometry and of its use in the MacDowell-Mansouri description of gravity. This material is not new, and partly overlaps with our previous expositions of Cartan geometry in physics \cite{derekmacd,dereksigma}.   While Cartan geometry is a broad subject (see e.g.\ \cite{Sharpe}), we include here a brief introduction sufficient for our main goal: the construction of observer space geometries.  We find it best to do this before leaving the familiar world of spacetime physics. 

\subsection{Model spacetimes}
\label{Klein}

Klein geometry studies homogeneous spaces via their symmetry groups.  If a Lie group $G$ acts transitively on a manifold $\GH$, we can identify $\GH$ with the coset space $G/H$, where $H$ is the stabilizer of an arbitrarily chosen point $z\in Z$:  
\ben
   H = \{g\in G : g\gh = \gh\}\,,
\een
a topologically closed (and hence Lie) subgroup.  The isomorphism $\GH\cong G/H$ is $G$-equivariant.  Conversely, the coset space $G/H$ is a manifold with smooth $G$ action, provided $H$ is closed in $G$.  It is thus convenient to define a Klein geometry to be such a pair of groups, even though we think of these as algebraic tools for studying the geometry of the corresponding homogeneous space. 
\begin{defn}
A \define{Klein geometry} $(G,H)$ is a Lie group $G$ with closed subgroup $H$.  
\end{defn}

For spacetime geometry, the obvious Klein geometries are the standard family of maximally symmetric solutions of the vacuum Einstein equations: de Sitter, Minkowski, or anti-de Sitter, depending on the cosmological constant $\Lambda$. 
In $3+1$ dimensions, the corresponding groups of isometries preserving both orientation and time orientation are
\ben
\label{model-spacetimes}
  G = \left\{ 
  \begin{array}{rlc}
  \SO_o(4,1) & \text{\sf de Sitter} & (\Lambda > 0)\\
  \ISO_o(3,1) & \text{\sf Minkowski} & (\Lambda = 0) \\
  \SO_o(3,2) & \text{\sf anti-de Sitter} & (\Lambda < 0)
  \end{array}
  \right. \qquad \quad
  H = \SO_o(3,1) \quad \text{\sf Lorentz group}.
\een
where the subscript ${}_o$ denotes the connected component.  Throughout the rest of this paper, unless otherwise noted, the letters $G$ and $H$ will refer to the particular groups in (\ref{model-spacetimes}), and $\GH$ will denote the corresponding homogeneous spacetime, where $H$ is the stabilizer of an arbitrary $\gh\in \GH$, fixed once and for all.  Since little of what follows depends on the sign of $\Lambda$, we treat all cases in parallel, noting exceptions as necessary.

In any of these spacetime Klein geometries, $G$ acts irreducibly on its Lie algebra $\g$ via the adjoint representation, but restricting this to the group $H$, we have an invariant direct sum
\ben
      \g = \h \oplus \ggh
\label{split1}
\een
where $\h$ is the Lie algebra of $H$, and  the complement $\ggh \cong \R^{3,1}$ may be identified in a canonical way with the tangent space at $z$:
\ben\ggh = T_{\gh}\GH.\een
Geometrically, this direct sum breaks the `infinitesimal symmetries' of $\GH$ up into `infinitesimal Lorentz transformations', which preserve $\gh$, and `infinitesimal translations' of $\gh$.   In fact, this gives a $\Z/2$-grading of $\g$ with even part $\h$ and odd part $\ggh$:
\ben
     [\h,\h] \subseteq \h \quad [\h,\ggh] \subseteq \ggh \quad [\ggh,\ggh] \subseteq \h
\een
A reductive geometry in which $\h\oplus\ggh$ is a $\Z/2$-grading is called a \define{symmetric space} \cite{helgason}.

In fact, even the metric on these homogeneous spacetimes can essentially be recovered from Lie theory.  Any semi-Riemannian metric on $\GH$ invariant under $G$ is induced by some nondegenerate $H$-invariant symmetric bilinear form on $\ggh$.  In each of the spacetime geometries $\GH$, there is only one such invariant bilinear form up to scale, so the geometry of $\GH$ reduces to knowing the groups $G$ and $H$, plus a unit of length. 

\subsection{Cartan geometry of spacetime}
\label{stcartan}

To each Klein geometry, there is an associated type of Cartan geometry.  Here we give the general definition, before specializing to the spacetime Klein geometries just discussed.
\begin{defn}
\label{def:cartan}
A {\bf Cartan geometry} $({\pi\maps P\to \M},A)$ modeled on the Klein geometry $(G,H)$ is a principal right $H$ bundle $\pi\maps P \to \M$ equipped with a $\g$-valued 1-form $A$ on $P$
\ben
        A\maps TP \to \g
\een
called the \define{Cartan connection}, satisfying three properties:
\begin{C-list}
\item \label{cartan:nondeg} For each $p\in P$, $A_p\maps T_p P\to \g$ is a linear isomorphism;

\item \label{cartan:equivariant} $(R_h)^\ast A = \Ad(h^{-1})\circ A \quad \forall h \in H$;
\item \label{cartan:mc} $A$ restricts to the Maurer--Cartan form on vertical vectors. 
\end{C-list}
\end{defn}

To be more precise in property \ref{cartan:mc}, note that we can pull back forms on $H$ along any local trivialization $f\maps P|_U \to U\times H$.  On each fiber $P_x$, this gives an isomorphism of {\em right} $H$-spaces, but only the pullback of the {\em left}-invariant Maurer--Cartan form on $H$---defined by $A_{H}(v)=(L_{h^{-1}})_\ast v$ for $v\in T_hH$---is independent of which trivialization we use.  This is the Maurer--Cartan form on $P$ that property \ref{cartan:mc} refers to.  

The \define{curvature} of a Cartan connection is the $(P \times_H \g)$-valued 2-form
\ben
    F = dA + \half[A,A]
\een
and a Cartan geometry is called \define{flat} if $F=0$. 
\begin{exple} 
Any Klein geometry $(G,H)$ becomes a Cartan geometry in a canonical way.  The map $G\to G/H$ is principal right $H$-bundle, and the Maurer--Cartan form $A\maps TG \to \g$ is a Cartan connection.  This Cartan geometry is flat by the Maurer--Cartan equation.  Conversely, any {\em flat} Cartan geometry is locally isomorphic to a Klein geometry.  (See e.g.\ {\rm \cite{Sharpe}}.)
\end{exple}

\subsubsection*{From Lorentzian geometry to Cartan geometry}
\label{lorentz-to-cartan}

We now focus on spacetime Cartan geometry and, in particular, give the precise correspondence to Lorentzian geometry in the familiar sense.

Fix any one of the three spacetime Klein geometries $(G,H)$ defined in section \ref{Klein}.  Starting with any manifold $\M$ equipped with a Lorentzian metric, orientation and time orientation, we will canonically construct a Cartan geometry associated to it.    

As the principal bundle, we take the bundle of all ways to glue $\GH$ to $M$ by identifying $\ggh=T_\gh\GH$ with some tangent space of $M$ in a way that respects all of the relevant structure on these tangent spaces.  More precisely, we define:
\ben
\label{frame-bundle}
 \fb =\{  \text{proper linear isometries }  f\maps \ggh \to T_x\M,\; x\in \M  \}\,,
\een
where we call a map \define{proper} if it preserves orientation and time orientation.  $\fb$ is isomorphic to the usual oriented and time oriented orthonormal frame bundle, so we refer to it as the \define{frame bundle} and to its elements as \define{frames}.  Since $H$ acts on $\ggh$, it acts on frames via composition, $f \mapsto f\circ h$, making $\fb$ into a principal right $H$ bundle.

The Cartan connection is a certain $\g$-valued 1-form on the bundle $\fb$, but since $\fb$ is a principal $H$ bundle, $H$-invariance of the splitting (\ref{split1}) means a $\g$-valued 1-form is simply an $\h$-valued 1-form together with a $\ggh$-valued 1-form.  In fact, we have a canonical $\ggh$-valued 1-form on $\fb$: 
\ben
     e\maps T\fb \to \ggh
\een
called the \define{soldering form}.  Given the canonical maps 
\ben
   \xymatrix@C=-.5em{
   &  T\fb \ar[dl]_{\varpi} \ar[dr]^{\pi_*} \\
   \fb & & T\M
   }
\een
where $\pi_*\maps T \fb \to T\M$ is the differential of the projection $\pi\maps\fb\to \M$ and $\varpi$ maps $v\in T_f \fb$ to $f$, the soldering form is given by
\ben
\label{soldering}
     e(v) = \varpi(v)^{-1}(\pi_*(v)).
\een

Moreover, the metric has a canonical Levi-Civita connection: the unique torsion-free connection on $T\M$. This corresponds to a torsion-free connection $\omega$ on $\fb$.  Together, we have 1-forms $\omega\maps T\fb \to \h$ and $e\maps T\fb \to \ggh$, which assemble to give:
\ben
   A\maps T\fb \to \g
\een
unique up to gauge transformations of the principal $H$ bundle $\fb\to \M$.  One may check that $(\pi\maps \fb \to M, A)$ is a Cartan geometry modeled on $(G,H)$.

\subsubsection*{From Cartan geometry back to Lorentzian geometry}
\label{fake}

We have just seen how to get a Cartan geometry starting from a Lorentzian spacetime, for any of the models $(G,H)$ from section \ref{Klein}.  To describe general relativity as a gauge theory for a Cartan connection, however, we need just the opposite: we want to see how an $H$-bundle with Cartan connection gives us a Lorentzian geometry.  

In a general Cartan geometry $({\pi\maps \ff\to \M},A)$ modeled on the spacetime geometry $(G,H)$, we can still think of the bundle $\ff$ as a stand-in for the bundle (\ref{frame-bundle}) of oriented orthonormal frames.  We often call it a \define{fake frame bundle}.  From it, we can construct the associated vector bundle 
\[
    \fake = \ff \times_H\ggh
\]
which we call the corresponding \define{fake tangent bundle}.  

An important observation is that the fake tangent bundle inherits a {\em metric} from the Klein geometry $(G/H)$:  since $\eta$ is an $H$-invariant inner product on $\ggh$, it induces a metric on $\fake$.  It should cause little confusion if we also call this metric $\eta$.   Since $\fake$ is only the {\em fake} tangent bundle, we cannot use the metric $\eta$ to measure lengths and angles for tangent vectors to $M$.  However, we can transfer this metric to $TM$ if we have a \define{coframe field}---a vector bundle isomorphism:
\[ 
\xymatrix @!0
{ T\M
 \ar [ddr]
  \ar[rr]^{\displaystyle e} 
  & &  
  \fake
  \ar[ddl]
  \\  \\ &
   \M 
 }
\]
This is the global analog of the local coframe fields $e\maps T\M \to \R^{3,1}$ often used in gravitational theory, and induces a metric on spacetime via pullback:
\ben
g(v,w)=\eta( e(v),e(w))\,.
\label{defmet}
\een
   Most importantly for our purposes, a coframe field can be obtained as part of a $\g$-valued 1-form on a principal $H$ bundle over $\M$:

\begin{lemma}
\label{coframe-lemma}
Let $p\maps \ff\to \M$ be a principal $H$-bundle, $\fake$ the corresponding fake tangent bundle.  Then there is a canonical one-to-one
correspondence between:
\begin{itemize}
\item vector bundle morphisms $e\maps T\M \to \fake$, and
\item $\ggh$-valued 1-forms $\varepsilon$ on $\ff$ that are:
 \begin{itemize}
 \item horizontal: $\varepsilon$ vanishes on $\ker(dp)$
 \item $H$-equivariant: $R^*_h \varepsilon = h^{-1} \circ
 \varepsilon$ for all $h\in H$.
 \end{itemize}
\end{itemize}
Moreover, the first of these is an isomorphism precisely when the
second is nondegenerate, meaning that each restriction $\varepsilon
\maps T_f\ff \to \ggh$ has maximal rank.

\end{lemma}

\begin{proof}
The proof is straightforward.  See \cite{tele} for details.  
\end{proof}

From this lemma, it is immediate that the $\ggh$ part of the Cartan connection is equivalent to a coframe field.  In particular, $\ff$ is isomorphic to the frame bundle (\ref{frame-bundle}) for the metric (\ref{defmet}).  
The $\h$ part of the Cartan connection is then an Ehresmann connection on $\ff$, which corresponds to a metric-compatible connection.

Note that `fake frame bundles' and `fake tangent bundles' are equivalent.  While we started with the principal bundle $\ff$ and built the vector bundle $\fake = \ff\times_H \ggh$, we could just as well start with an arbitrary `fake tangent bundle' $\fake$---a vector bundle $\fake\cong TM$ equipped with a metric $\eta$ and an orientation, and mimic the construction of $\fb$ to get the principal $H$ bundle
\ben
\label{fake-frame-bundle}
 \ff \cong\{  \text{proper linear isometries }  f\maps \ggh \to \fake_x,\; x\in \M  \}\,.
\een

\subsection{MacDowell--Mansouri gravity}
\label{sec:macdo}

The formulation of general relativity in which spacetime Cartan geometry plays the most conspicuous role is the action introduced by MacDowell and Mansouri \cite{macdo}.  MacDowell--Mansouri gravity works only with a nonzero cosmological constant, so here we take $(G,H)$ to be either the de Sitter or anti-de Sitter model.  Fixing a fake frame bundle over spacetime $\M$, the only field in the theory is a Cartan connection $A$.   The action is
\ben
S_{{\rm \scriptscriptstyle MM}}[A] = \int_\M \kappa_{\h} \left(F_\h\wedge F_\h\right)\,,
\label{macdow}
\een
where $F=F[A]$ is the curvature, $F_\h$ is its $\h$-valued part, and $\kappa_{\h}$ is a non-degenerate $H$-invariant inner product on $\h$. There is a two-parameter family of such products, of the form
\ben
\kappa_{\h}(X,Y) = \tr_{\h}\left(X(c_0 + c_1\star)Y\right),\quad X,Y\in\h,
\label{forminner}
\een
where $\star$ is a Hodge star operator on $\h\cong \Lambda^2\R^{3,1}$, and $\tr_{\h}$ is the Killing form on $\h$. 

To see how this gives general relativity, first note that the splitting $\g = \h \oplus \ggh$ reduces $A$ to an $H$ connection $\om$ together with a coframe field $e$.  The $\h$ part of the curvature is then
\[
    F_\h = R+ \half[e,e]
\]
where $R=R[\om]$ is the curvature of $\om$.  Substituting this into the action, we get:
\bena
S_{{\rm \scriptscriptstyle MM}}[A] 
&= &c_0\int_\M \tr_{\h} \left([e,e]\wedge R
+ \frac14 [e,e]\wedge [e,e]
+ R\wedge R\right)\nonumber
\\&&+ c_1\int_\M \tr_{\h} \left([e,e]\wedge {\star R} 
+ \frac14 [e,e]\wedge {\star [e,e]}
+ R\wedge \star R\right)\,.
\label{macdowsub}
\eena

In the second line, with appropriate normalization, the first two terms are just the Palatini action for general relativity; the final term is a topological `Gauss--Bonnet' term, which does not affect the classical theory.  In fact, the resulting field equations for general relativity can also be neatly summarized as 
\ben
   [e,\star F] = 0 
\label{gravityeq}
\een  
where $\star$ acts only on the $\h$ part of the full curvature $F$. 
In the first line, $\tr_{\h}(R\wedge R)$ is also topological (the `Hirzebruch signature'), while the second term vanishes under the trace (see \cite{dereksigma}). The first term is then the `Holst term' added to the Palatini action to introduce the `Immirzi parameter' $\gamma=c_1/c_0$, and also does not modify the vacuum equations of motion (\ref{gravityeq}).

It is worth noting that while $G$ symmetry is broken `by hand' to $H$ in this action, the symmetry can also be broken `spontaneously', as first observed by Stelle and West \cite{stellewest}.   Discussion of this idea and its geometric significance can also be found in \cite{dgr,broken}.  However, we will not need this here.   Much more detail on the Cartan geometric underpinnings of MacDowell-Mansouri gravity and related theories can be found in our previous work \cite{dgr,derekmacd,dereksigma}.

\subsection{Geodesics and development}
\label{sec:development}

One of the central statements of general relativity is that test particles move along geodesics for the Levi-Civita connection on spacetime.  In order to define the geodesics for a given Cartan geometry $(\pi:\ff\to M,A)$, the central notion is {\it development}, a mapping of paths in $M$ into the model space $\GH\cong G/H$.  We first define the development using paths in $\ff$, but then see immediately that it depends only on the projection of the path to $M$. 
\begin{defn}
Given a Cartan connection $A$ on $\ff$, the {\bf development} of a path $\mathcal{C}:[0,1]\to\ff$ is a path $\gh_{\mathcal{C}}:[0,1]\to G/H$ defined by $\gh_{\mathcal{C}}(t)=g_{\mathcal{C}}(t)H$, where $g_{\mathcal{C}}(t)$ solves the differential equation 
\ben
g_{\mathcal{C}}^{-1}\,\dot{g}_{\mathcal{C}}=A(\dot{\mathcal{C}})
\een
where $\dot{\phantom{x}}$ denotes differentiation with respect to $t$.
\end{defn}
\noindent Clearly $g_{\mathcal{C}}$ is unique up to choosing $g_{\mathcal{C}}(0)$, hence $\gh_{\mathcal{C}}$ is unique after fixing an origin $\gh_{\mathcal{C}}(0)\in G/H$, usually the identity coset.  Note that $g^{-1}\dot g$ is shorthand for $L_{g^{-1}\ast} \maps T_g G \to \g$ applied to $\dot g$. 

\begin{lemma}
\label{lem:geodes}
Let $\mathcal{C}$ and $\mathcal{C}'$ be two paths in $\ff$ with the same projection into the base manifold: $\pi\circ\mathcal{C}=\pi\circ\mathcal{C}'$.  Then the developments $\gh_{\mathcal{C}}$ and $\gh_{\mathcal{C}'}$ are equal.  
\end{lemma}
\begin{proof}
If $\pi\circ\mathcal{C}=\pi\circ\mathcal{C}'$, we can write $\mathcal{C}(t)=\mathcal{C}'(t) h(t)$ for some path $h$ in $H$.  If $g_{\mathcal{C}}$ is a solution to the equation $g_{\mathcal{C}}^{-1}\,\dot{g}_{\mathcal{C}}=A(\dot{\mathcal{C}})$,  then $g':=g_{\mathcal{C}} h$ solves
\ben
{g'}^{-1}\dot{g'} = h^{-1}A(\dot{\mathcal{C}})\dot{h} + h^{-1}\dot{h} 
\label{difeq}
\een
Now $\dot{\mathcal{C}'}=\dot{\mathcal{C}} h + \mathcal{C} \dot{h}$ and $A$ satisfies properties \ref{cartan:equivariant} and \ref{cartan:mc} in definition \ref{def:cartan}: $(R_h)^\ast A = \Ad(h^{-1})\circ A$ and $A$ restricts to the Maurer-Cartan form on vertical vectors. Hence the right-hand side of (\ref{difeq}) is $A(\dot{\mathcal{C}'})$, and the development of $\mathcal{C'}$ is $\gh_{\mathcal{C'}}=g'H=g_{\mathcal{C}}H=\gh_{\mathcal{C}}$.
\end{proof}

This lemma lets us define development of a path in $\M$ using an arbitrary lifting to $P$. Geometrically, the development $\gh_\mathcal{C}$ of $\mathcal{C}$ into $\GH$ represents the path traced out on the `model' space $\GH$ as it is rolled along $\pi\circ\mathcal{C}$ in the manifold $\M$. 

\begin{prop}
\label{chargeodes}
Let  $({\pi\maps P\to \M},A)$ be a Cartan geometry for one of the models $Z\cong G/H$ in {\rm (\ref{model-spacetimes})}, and let $A=\omega + e$ be the reductive splitting of the Cartan connection.  A path $\gamma$ in $\M$ is a geodesic for the connection on $TM$ induced by $\omega$ if and only if its development $\gh_\gamma$ is a geodesic in $\GH$. 
\end{prop}

\begin{proof}
We consider the case where the fake frame bundle $\ff$ is the actual frame bundle $FM$, and the Cartan geometry is the canonical one constructed in section \ref{lorentz-to-cartan}.  The general case is no harder, but requires translating between `fake' and `real' bundles using the coframe.  

Since $\omega$ is an Ehresmann connection on $FM$, $\gamma$ has a unique horizontal lift $\tilde\gamma$ to $FM$.  At each value of $t$ we have a frame $\tilde \gamma(t) \maps \ggh \to T_{\gamma(t)}M$, and the condition for parallel transport along $\gamma$ is that 
\[
   \tilde \gamma(t)^{-1} (\dot\gamma(t)) =: z \in \ggh
\]
is independent of $t$.   On the other hand, since $\tilde \gamma$ is horizontal, i.e. $\omega(\dot{\tilde\gamma}(t)) = 0$, the differential equation for calculating the development reduces to 
\[
      g(t)^{-1}\dot g(t) = e(\dot{\tilde\gamma}(t))
\]
But $e$ is the soldering form (\ref{soldering}), so the right hand side becomes $e(\dot{\tilde\gamma}(t)) = {\tilde\gamma}^{-1}(\dot\gamma(t)) = z$.  Thus we find that $\dot g(t) = g(t) z$, or more properly
\[
   \dot g(t) = L_{g(t)\ast} z. 
\]
That is, $g(t) = \exp(tz)$ for a constant $z\in \ggh$.  This is just the condition $t\mapsto g(t)H$ be a geodesic starting at the identity coset in the symmetric space $G/H$ (cf. Theorem IV.3.3 of Helgason \cite{helgason}).  This shows a geodesic in $M$ has geodesic development; for the converse, simply reverse this argument.  
\end{proof}

\sect{Observer space}
\label{sec:observerspace}

We now move on to the main goal of this paper: to understand general relativity in terms of {\em observer space}, rather than spacetime,  setting the geometric stage for modifications that could describe physics beyond general relativity.  We do this using the basic machinery of Cartan geometry described in section \ref{sec:spacetime}, which will allow us to define observer space geometries without relying on spacetime.  First, however, we describe in general terms what sort of geometric features observer space inherits from the geometry of a Lorentzian spacetime and its tangent bundle.  

\subsection{The observer space of a Lorentzian spacetime}
\label{spacetime-observer-spaces}

Given a Lorentzian spacetime $(M,g)$ equipped with a time orientation, we define the \define{observer space} $O$ of $\M$ to be the space of all unit future timelike vectors, also known as the future unit tangent bundle of $\M$.   We review here what sort of geometric features such an observer space naturally has: a contact structure, a Sasaki metric, and other canonical distributions.  Mathematically, this is standard material (see e.g.\ \cite{arnold,michorbook}); what is new here is the interpretation in terms of observer space.

\subsubsection*{Contact geometry} 

The observer space $O$ has a canonical {\em contact structure}.  Recall that on any $(2n+1)$-dimensional manifold, a contact form is a 1-form $\alpha$ that is maximally nonintegrable in the sense that the $(2n+1)$-form $\alpha \we d\alpha \we \cdots \we d\alpha$ is nowhere-vanishing, and hence is a volume form.   In contrast, note that the Frobenius integrability condition is that $\alpha\we d\alpha$ vanishes identically, so a contact form is indeed highly nonintegrable.  A \define{contact structure} on a manifold is the hyperplane distribution given by the kernel of a contact form.  The kernel does not change if we multiply the contact form by any nowhere-vanishing function, so we can also define a contact structure as an equivalence class of contact forms under multiplication by nonvanishing functions.  

In fact, given a Lorentzian manifold $\M$, its observer space $O$ has not only a canonical contact structure, but a canonical contact form inducing this structure.  This contact form, 
\[
          \alpha\maps TO \to \R
\]
is given by: 
\ben
\label{contact-structure}
\alpha(v) = g(p(v),\pi_* v)\,,\een where $p\maps TO\to O\subseteq T\M$ is the tangent bundle of $O$, and $\pi_*\maps TO\to T\M$ is the differential of the observer bundle $\pi \maps O \to \M$. 

The contact form $\alpha$ induces a \define{Reeb vector field} $\mathbf{r}$ on observer space, the unique vector field normalized by $\alpha$ and whose flow preserves $\alpha$, that is
\ben
    \alpha(\mathbf{r}) = 1 \qquad \text{and}\qquad \L_{\mathbf{r}}\alpha = 0\,,
\een
where $\L$ denotes the Lie derivative.   The Reeb vector field is the restriction to observer space of the `geodesic spray', the vector field on $TM$ whose integral curves give geodesics on $M$ for the Levi-Civita connection of $g$ (see e.g.\ \cite[Ch. V]{michorbook}).   Thus ``inertial observers'', who follow timelike geodesics in spacetime, simply follow the flow of the Reeb vector field in observer space. 

\subsubsection*{The Sasaki metric}

Besides being a contact manifold, the observer space of a Lorentzian spacetime is also naturally a semi-Riemannian manifold.  In fact, it has natural metrics of both Lorentzian and Riemannian signature.   

The key to this is that the double tangent bundle of any semi-Riemannian manifold $(M,g)$ has a pair of projections down to the tangent bundle:
\ben
   \xymatrix@C=-.5em{
   &  TT\M \ar[dl]_{\pi_*} \ar[dr]^{\kappa} \\
   T\M & & T\M
   }
\een
The first of these, $\pi_\ast$, is just the differential of the tangent bundle $\pi\maps TM\to M$, while the second, $\kappa$ is the `connection mapping' defined using the Levi-Civita connection $\nabla$ on $M$.  More precisely, since a vector field on $M$ is, in particular, a map $\xi\maps M \to TM$, its differential is $\xi_\ast\maps TM \to TTM$, and $\kappa$ is uniquely determined by requiring 
\ben
   \kappa(\xi_\ast X) = \nabla_X \xi
\een
for every vector field $\xi$ and every tangent vector $X$ to $M$ (see e.g.\ \cite{dida}).   One can show that $TTM$ is a direct sum of the distributions given by the kernels of these projections:
\[
     TTM = \ker \pi_* \oplus \ker \kappa 
\]
called respectively the \define{vertical} and \define{horizontal distributions} on $TM$.  Moreover,
\[
    \tilde g(v,w) = g(\pi_*v,\pi_*w) + g(\kappa v, \kappa w)
\]
is thus a nondegenerate metric on $TM$.  In our case, where $g$ has signature $(3,1)$, $\tilde g$ has signature $(6,2)$.  The observer space $O\subseteq TM$ is a submanifold with induced metric of Lorentzian signature $(6,1)$.  This induced metric is the \define{Sasaki metric} on observer space.

On the other hand, we can also get a {\em Riemannian} metric on observer space.  The Reeb vector field is a nonvanishing timelike vector field orthogonal to the symplectic structure under the Sasaki metric.  We can thus flip the sign of the Sasaki metric in the direction of the Reeb vector field, while fixing the metric on the contact distribution.  

\subsubsection*{Spatial, temporal and boost distributions} 

The contact structure on observer space is a hyperplane distribution.  But there are other canonical distributions on observer space as well.  First, there is the distribution complementary to the contact structure, the 1-dimensional distribution spanned by the Reeb vector field, which we call the \define{temporal distribution}.   Timelike geodesics on $M$ correspond to curves on $O$ tangent to the temporal distribution.  

But also, the 6-dimensional contact distribution splits into a pair of 3-dimensional subdistributions.  The \define{boost distribution} is the kernel of $\pi_*$, where $\pi\maps O \to M$ is the restriction of the tangent bundle to observer space.  A path tangent to this distribution corresponds to simply changing observers at a fixed point in spacetime.   By definition of the contact form (\ref{contact-structure}), the boost distribution is clearly a 3-dimensional sub-distribution of the contact distribution.  The \define{spatial distribution} can be defined either as the orthogonal complement of the boost distribution in the contact distribution, or as the kernel of  $\kappa$ restricted to $O$.  In this way, each tangent space splits canonically into a direct sum: 
\ben
\label{canonical-split}
  T_oO = (\text{boost vectors}) \oplus (\text{spatial vectors}) \oplus  (\text{temporal vectors})
\een
where the first two summands constitute the contact structure.  

\subsubsection*{Momentum space}

In spacetime physics, whereas the velocity of a subluminal particle is a unit future-timelike vector, the {\em momentum} is an arbitrary future-timelike cotangent vector.  The space of momenta is thus a subspace of $T^\ast\!\M$. A particle's momentum can be dualized and then normalized via the metric to get the corresponding velocity, and conversely, the dual of the velocity times the particle's mass gives the momentum.

In observer space, on the other hand, the velocity of a particle is not a vector but just a point.  However, we can still reconstruct the space of all momenta of particles directly from observer space, without appealing to spacetime in any direct way: it is the `symplectification' of observer space.  Any $(2n+1)$-dimensional contact manifold extends canonically to a $(2n+2)$-dimensional symplectic manifold (see e.g.\ \cite{arnold}).

To recall how this works, first define a \define{contact element} at a point $o$ in a contact manifold $O$ to be a covector $\beta_o \in T^\ast_oO$ whose kernel is the contact hyperplane at $o$.  Then the \define{symplectification} $S$ of $O$ is the space of all contact elements on $O$.  In other words, $S$ is the subbundle of $T^\ast O$ whose sections are contact forms for the contact structure.  The symplectic structure on $S$: 
\[
       \om \maps TS \times TS \to \R
\]
is the differential of the canonical 1-form $\alpha \maps TS \to \R$ given by:
\[
    \alpha(v) = \beta_o(\pi_\ast v)  \qquad  v\in T_{\beta_o} S
\]
where $\pi\maps S \to O$ is the obvious bundle and $\beta_o$ is a contact element at $o\in O$. 
 
While the contact distribution alone determines the symplectification, in the cases of interest here we always have a specified contact form, and this gives a preferred section of the symplectification.  Physically speaking,  this section selects, at each point of observer space, the momentum corresponding to a particle with unit mass.   So far, we have defined $O$ as the unit future tangent bundle of a Lorentzian spacetime; in this case the contact form (\ref{contact-structure}) provides the normalizing section.  
 
\subsubsection*{Lightlike particles} 

At this point, we may wonder how massless particles fit into the observer space picture.  Since a point in observer space corresponds to the instantaneous velocity of a subluminal particle, there seems to be no room in observer space for particles traveling at the speed of light.    The answer, however, is clear: lightlike particles live in the `boundary' of observer space.  

More precisely, each fiber in the observer space bundle $O \to M$ is a copy of hyperbolic space, namely the space of all timelike velocity vectors at the same spacetime event.  If we adjoin to each fiber the space of lightrays at that event, we obtain an extension of observer space in which each fiber is a compactification of hyperbolic space, diffeomorphic to a 3-dimensional ball.   

This extended observer space, which includes `lightlike observers' on its boundary, deserves further study.  However, for the remainder of this paper, we focus on ordinary subluminal observers only.

\subsection{Model observer spaces}
\label{modelob}

To describe observer space Cartan geometry, we first need Klein geometries to serve as homogeneous models of observer space.  Fortunately, there are three obvious choices: The model spacetimes $\GH\cong G/H$ described in section \ref{Klein} are not only homogeneous, but also \define{isotropic}, meaning that $G$ acts transitively on observers as well.  Thus, to each of these models there is a corresponding model observer space.

Recall that our chosen spacetime event, $z\in Z$, has stabilizer $H$ and $\ggh = T_z \GH$.  The observers at $z$, the unit timelike vectors, are thus elements of 
\define{hyperbolic 3-space}, which we define to be a submanifold of $\ggh \cong \R^{3,1}$ in the usual way:
\ben
    \Hyp^3 := \{ y \in \ggh : \eta(y,y) = -1,\; y_0 >0\}\cong H/K\,,
\een
where we are using a $G$-invariant metric $\eta$ on $\GH$ of signature $({-}{+}{+}{+})$, and the stabilizer of $y \in \Hyp^3$ is
\ben
   K \cong \SO(3)\,. 
\een  
The natural projection from observer space down to spacetime:
\ben
\label{GK->GH}
     G/K \to G/H
\een
is a $G$-equivariant fiber bundle with standard fiber $\Hyp^3$. The fiber over an event is of course just the hyperbolic space of all velocities that observers at that event can have. 

As representations of $K$, both $\h$ and $\ggh$ in the reductive splitting (\ref{split1}) are further reducible, each splitting into a direct sum of two irreducible $K$ representations: 
\ben
    \h = \k \oplus \ghk \qquad \text{and}\qquad   \ggh = \vec \ggh \oplus \ggh_o\,.
\label{split2}
\een
Here $\k$ is the Lie algebra of $K$, and the complement $\ghk \cong \R^3$ is canonically isomorphic to  the tangent space to $H/K$ at the basepoint $\hk$.    The representation $\ggh$, corresponding to spacetime translations, naturally splits into spatial translations $\vec\ggh$ and temporal translations $\ggh_o$, from the observer's perspective. 
 
Thus the adjoint representation of $K\subseteq G$ on $\g$ splits into a direct sum of four irreducible representations, based on our choice of basepoint $\gh\in \GH$ and observer $\hk$ at $\gh$:
\ben
\label{split3} 
      \g = \k \oplus (\ghk \oplus \vec \ggh \oplus \ggh_o) 
\een
The parenthesized part is naturally identified with the tangent space to observer space at the chosen observer; the three summands correspond respectively to those in the canonical splitting (\ref{canonical-split}) of a tangent space to observer space.  

We can interpret the sum (\ref{split3}) of $K$ representations in terms of infinitesimal symmetries of observer space: 
\[
  \begin{array}{lcl}
    \k &\sim& \text{rotations around the observer} \\
    \ghk &\sim& \text{boosts, changing the observer but not the base event}\\  
    \vec \ggh &\sim& \text{spatial translations of the event/observer}\\  
    \ggh_o &\sim& \text{time translations of the event/observer.}\\  
  \end{array}
\]
It is a straightforward exercise to work out the Lie brackets.   For any of the models, we have 
\ben
\begin{array}{llll}
{}\makebox[2.2em]{$[\k, \k]$} \subseteq \k 
  & \makebox[2.2em]{$[\ghk,\ghk]$} \subseteq \k 
  & \makebox[2.2em]{$[\vec\ggh,\vec\ggh]$} \subseteq \k 
  & \makebox[2.2em]{$[\ggh_o,\ggh_o]$} =0\\
{}\makebox[2.2em]{$[\k, \ghk]$} \subseteq \ghk 
  & \makebox[2.2em]{$[\ghk,\vec\ggh]$}  \subseteq \ggh_o
  & [\vec\ggh,\ggh_o]  \subseteq \ghk  \\
{}\makebox[2.2em]{$[\k, \vec\ggh]$} \subseteq \vec\ggh 
  & \makebox[2.2em]{$[\ghk,\ggh_o]$}  \subseteq \vec\ggh \\
{}\makebox[2.2em]{$[\k, \ggh_o]$}  =0 
\end{array}
\label{commutators}
\een
where the third column is actually zero in the case $\Lambda = 0$.  The geometric interpretation of these is clear: for example, $[\ghk,\vec\ggh]  \subseteq \ggh_o$ says tiny boosts and tiny spatial translations commute up to a time translation.  

From the above chart we see that $[\k,\ghk\oplus \ggh]\subseteq \ghk\oplus \ggh$, but $ [\ghk\oplus \ggh,\ghk\oplus \ggh]$ generally has parts in both $\k$ and its complement.  Thus our model observer geometries are reductive but not symmetric.   Note that $\g = \k \oplus \ghk \oplus \vec \ggh \oplus \ggh_o$ makes $\g$ into a $(\Z/2 \times \Z/2)$-graded Lie algebra. 

Another geometric interpretation of the splitting (\ref{split3}) will be important for understanding `Cartan geometrodynamics' in section \ref{cartangeo}.   Cartan geometrodynamics is about how observer space geometry relates spacetime geometry to spatial geometry.  While general relativistic spacetimes have no canonical notion of `space', our homogeneous models of spacetime do.  In particular, given an observer, there exists a unique maximal totally geodesic hypersurface orthogonal to the observer.  The stabilizer of such a hypersurface is a subgroup $H'\subseteq G$, which acts transitively on the hypersurface.  The stabilizer of a point on this spatial slice under the action of $H'$ is the same as the stabilizer of an observer, namely $K$.  This lets us describe {\em spatial} geometry using these groups:
\[
  H' \cong \left\{ 
  \begin{array}{rl}
  \SO(4) & \text{\sf spherical group} \\
  \ISO(3) & \text{\sf Euclidean group} \\
  \SO(3,1) & \text{\sf hyperbolic group}
  \end{array}
  \right. \qquad \quad
  K = \SO(3) \quad \text{\sf rotation group.} 
\]

From this perspective, we get a rather different interpretation of the decomposition of $\g$ in (\ref{split3}):
\[
  \begin{array}{lcl}
    \k &\sim& \text{rotations around the basepoint on the spatial slice} \\
    \ghk &\sim& \text{changes of spatial slice, without changing the basepoint}\\  
    \vec \ggh &\sim& \text{translations of the basepoint, within the same spatial slice}\\  
    \ggh_o &\sim& \text{time translations of the basepoint, changing the spatial slice.}\\  
  \end{array}
\]

It is worth mentioning that the models discussed in this section are not the only possible models of observer space.  For example, starting from one of the models $(G,K)$, one could form a new model $(G',K)$ where $G'=K\ltimes (\ghk\oplus \ggh)$.   This example---which is nothing but the observer space of Galilean spacetime---is a `mutation'  of the models derived from homogeneous Lorentzian spacetimes,  meaning that the Lie algebra $\g'\cong \g$, not as Lie algebras but as representations of $K$.   Mutations carry the same essential geometric information \cite{Sharpe}.  In this paper, we use only the three models $(G,K)$ described above.

\subsubsection*{Contact structure}

The homogeneous models of observer space are homogeneous as contact manifolds, meaning that the symmetries preserve the contact structure.  We now describe the contact structure directly in the language of Kleinian geometry. 

\begin{prop}
Consider the observer space Cartan geometry $(\pi \maps G \to G/K, A)$ where $A$ is the Maurer--Cartan form.   The projection of $A$ into $\ggh_o$:
\[
\xymatrix{
  TG \ar[r]^A & \g \ar[r] & \ggh_o 
  }
\]
induces a 1-form on $G/K$.  This 1-form coincides with the negative of the contact form (\ref{contact-structure}) induced by the structure of $G/K$ as the unit future tangent bundle of $G/H$. 
\end{prop}
\begin{proof}
First, by a slight variation of Lemma~\ref{coframe-lemma}, the projection of the Maurer--Cartan form into $\ghk \oplus \ggh$ is the same as an isomorphism $$T(G/K) \to G\times_K (\ghk \oplus \ggh).$$  Similarly, projecting further to $\ggh_o$, the $\ggh_o$ part of $A$ is the same as a 1-form on $G/K$ with values in the associated vector bundle $G\times_K \ggh_o$ over $G/K$. But this latter vector bundle is trivial, since $K$ acts trivially on $\ggh_o$, so we have just a $\ggh_o$-valued 1-form.   Taking advantage of the above isomorphism for $T(G/K)$,  this 1-form is given by:
\[
     \begin{array}{ccc}
   G\times_K (\ghk \oplus \ggh) & \to & \ggh_o \\
   {}[g,(a,w)] & \mapsto & w_o
     \end{array}
\]
where $w_o$ is the $\ggh_o$ part of $w\in \ggh$.  Namely, using the metric $\eta$ on $G/K$, 
\[
   w_o = - \eta(\hk,w) \hk\,
\]
since $\ggh_o$ is by definition the span of $\hk\in \ggh$. Using the isomorphism $\ggh_o \cong \R$ given by $\hk\mapsto 1$,  we thus get a real-valued 1-form $\tilde \alpha$ on $G/H$:
\[
     \tilde\alpha([g,(a,w)] ) = - \eta(\hk,w).
\]

On the other hand, note that the differential of the bundle map (\ref{GK->GH}) at the observer $\hk$ (corresponding to the identity coset) is just the obvious projection
$
 \ghk\oplus \ggh \to \ggh. 
$ 
Thus, the contact form (\ref{contact-structure}) is given by:
\[
     \begin{array}{rccc}
    \alpha_\hk \maps\!\!\!\!& G \times_K (\ghk \oplus \ggh)& \to & \R \\
  & [g,(a,w)] & \mapsto & \eta(\hk,w)
     \end{array}
\]
This differs from $\tilde \alpha$ by a minus sign.  
\end{proof}

For the momentum space corresponding to the model observer space, the symplectification of the contact manifold $G/K$ is canonically isomorphic to the space of future-timelike cotangent vectors to $G/H$.

\subsection{Cartan geometry of observer space}
\label{observer-cartan}

Now that we understand the homogeneous models, the general definition of Cartan geometry (def.~\ref{def:cartan}) lets us define abstractly what it means for a manifold to have the geometry of an `observer space'. 
\begin{defn}
\label{def:obgeo}
An \define{observer space geometry} $({\pi\maps P\to O},A)$ is a Cartan geometry modeled on $(G,K)$, where $G$ is $\SO(4,1)$, $\ISO(3,1)$ or $\SO(3,2)$ and $K$ is $\SO(3)$, as described in section \ref{modelob}.  
\end{defn}
This definition is intrinsic in the sense that it describes the geometry of an observer space directly, without using any underlying notion of `spacetime'.  We now turn to describing observer space Cartan geometries in more detail, and give some examples. 

\subsubsection*{Geometry of observer space connections} 

According to the splitting (\ref{split3}) of the Lie algebra $\g$, the Cartan connection {on observer space} breaks up into four irreducible pieces:
\ben
\begin{array}{ccccccccc}
 \g & = & \k &\oplus & \ghk & \oplus & \vec\ggh & \oplus & \ggh_o \\
 A & = & \Omega & + & b & + & \vec e & + & e_o
\end{array}
\label{irredpie}
\een
This can be interpreted geometrically using `rolling without slipping':  Locally, as we begin to move along some path $\gamma$ in $O$, rolling the model observer space along as we go, the transformation of the model space breaks up into:
\begin{itemize}
\item \parbox{\widthof{$e_o(\gamma'(0))$}}{$\Omega(\gamma'(0))$} ---a tiny rotation around the model observer
\item \parbox{\widthof{$e_o(\gamma'(0))$}}{$b(\gamma'(0))$} ---a tiny boost of the model observer
\item \parbox{\widthof{$e_o(\gamma'(0))$}}{$\vec e(\gamma'(0))$} ---a tiny spatial translation of the model observer
\item $e_o(\gamma'(0))$ ---a tiny time translation of the model observer
\end{itemize}
Using the commutation relations of the algebra, the curvature is:
\ben
\label{curvature-split} 
\textstyle
\begin{array}{ccl} 
     F &=& dA + \half[A,A] \\[.5em]
     &=& (d\Om + \half[\Om,\Om] + \half[b,b] + \half[\vec e, \vec e] )
      + (d_\Om b +[\vec e, e_o]) 
      + (d_\Om \vec e + [b,e_o]) 
      + (d e_o + [b,\vec e]),
\end{array}
\een
where the four parenthesized terms live respectively in $\k$, $\ghk$, $\vec\ggh$, and $\ggh_o$.

\subsubsection*{Distributions} 

In Cartan geometry based on any model $(G,K)$, a $K$-invariant structure on the tangent space $\g/\k$ of the model Klein geometry yields the same type of structure on tangent spaces of the model.  For an observer space Cartan geometry, according to (\ref{split3}) we have a $K$-invariant splitting
\[
     \g/\k \cong \ghk \oplus \vec\ggh \oplus \ggh_o
\]
and the Cartan connection thus gives a splitting of each tangent space.  To see how this works, first note that the Cartan connection gives four distributions on the total space $P$, just by taking preimages of the components along $A\maps TP \to \g$.  It will be convenient to denote the distribution corresponding to a particular subalgebra using an underline, for example:
\[
\begin{array}{ccl}
   \uk&=& A^{-1}(\k) \\
   \ughk &=& A^{-1}(\ghk) \\
   \uh &=& \uk \oplus \ughk = A^{-1}(\h) \\ & \vdots    
\end{array}
\]
Note that by definition $A(\underline{\mathfrak{q}}) = \mathfrak{q}$ for any $K$-invariant subspace $\mathfrak{q}\subseteq \g$.

Because of $K$-invariance, these distributions descend to observer space in a gauge invariant way.  The distribution $\uk$ is just the vertical distribution of the bundle, and so is trivial on the base; the distributions corresponding to $\ghk$, $\vec \ggh$ and $\ggh_o$ descend to give the local directions corresponding to boosts, spatial translations, and time translations, according to each observer.

\subsubsection*{Observer space geometry from spacetime geometry}

 Of course, the most obvious way to construct an observer space geometry is to start with a spacetime geometry.  Given a spacetime Cartan geometry, the principal $H$ bundle $\ff$ over spacetime determines an isomorphism between the (fake) observer space $\fo$ and the associated bundle $P\times_H H/K$. Defining the fake tangent bundle $\fake = \ff \times_H\ggh$ and viewing $P$ as the fake frame bundle
\ben
 \ff =\{ \text{linear isometries }  f\maps \ggh \to \fake_x,\; x\in \M \}\,,
\een
we can define \define{fake observer space} $\fo$ as the bundle of unit future-directed timelike vectors in the fake tangent bundle. We have a canonical bundle isomorphism
\ben
\begin{array}{ccc}
    \ff \times_{H} H/K & \to & \fo \\
  {} [f_x,\hk] &\mapsto & f_x(\hk)\,.
\end{array}    
\een
Starting from $\ff$, we can fix an observer $\hk\in H/K$ which gives us a projection map
\ben
\begin{array}{ccc}
    \ff & \to & \fo \\
   {} f_x &\mapsto & f_x(\hk)
\end{array}    
\een
allowing us to identify $\ff$ as a principal $K$ bundle over $\fo$. The original Cartan connection on spacetime becomes a Cartan connection on observer space:
\begin{lemma}[Observer space geometry from spacetime geometry]
\label{lem:cart}
If $(\pi\maps P \to \M,A)$ is a Cartan geometry with model $(G,H)$, then $(\pi\maps P \to \fo,A)$, where $\fo=P\times_H H/K$ is a Cartan geometry with model $(G,K)$.
\end{lemma}

\begin{proof}
First note that a Cartan connection on $\ff\to\M$ and a Cartan connection on $\ff\to\fo$ are both $\g$-valued 1-forms on $\ff$. If $A$ is a Cartan connection on $\ff\to\M$, it has properties \ref{cartan:nondeg}-\ref{cartan:mc} in definition \ref{def:cartan}: It is a linear isomorphism $T_p\ff\to\g$ at each $p\in\ff$, it transforms under the adjoint of $H$, and is the Maurer-Cartan form on invariant vector fields associated to the Lie algebra $\h$. But the second and third properties imply the same properties for the subgroup $K\subseteq H$ and the subalgebra $\k$. Hence $A$ is also a Cartan connection on $\ff\to\fo$.
\end{proof}

\noindent The converse is of course not true: the $K$ action on a principal $K$ bundle need not extend to an $H$ action, and even if it does, properties \ref{cartan:equivariant} and \ref{cartan:mc} for the group $K$ do not imply those for $H$.

As explained in section \ref{stcartan}, a Cartan connection on the bundle $\ff\to\M$ allows us to reconstruct the `real' bundles from the fake ones, since its $\ggh$ part defines a coframe and hence an isomorphism between the bundles $TM$ and $\fake$.

\subsection{Observer fields}
\label{obsfieldsec}

From the observer space perspective, the geometry of {\em spacetime} is always viewed locally in relation to some particular observer.  To study spacetime geometry over extended regions, it is thus helpful to single out one observer at each point by specifying a `field of observers'. As in our previous discussions we will start with a Lorentzian spacetime with its associated Cartan connection on the frame bundle $\fb$. 
\begin{defn} 
If $M$ is a Lorentzian spacetime with observer space $O$, a \define{field of observers} is a section of the bundle $O\to M$. 
\end{defn}
\noindent In other words an observer field is a unit future-directed timelike vector field.  These exist on any time-oriented Lorentzian manifold. 

An observer field $u$ gives a vector field on $\M$, and hence a 1-form $\hat u$, defined by $\hat{u}(v)= g(u,v)$.  However, we can also obtain this dual 1-form by pulling back the contact form $\alpha$ along $u$:
\begin{prop}
The 1-form dual to $u$ is $\hat u = u^*\alpha$. 
\end{prop}
\begin{proof}
Directly calculating the pullback of $\alpha$ along $u$, we get: 
\begin{align*}
u^\ast \alpha(v) &=\alpha(u_\ast v) \\
&= g(p(u_\ast v), \pi_\ast u_\ast v)\\
&=g(u,v)\\
&=\hat u(v),
\end{align*}
since by definition $\hat u$, the dual of $u$, is the 1-form on $\M$ given by
$   \hat u(v) = g(u,v)$. 
\end{proof}

A field of observers $u$ and its corresponding field of \define{co-observers} $\hat{u}$ give us a canonical way to split differential forms into spatial and temporal parts.   First \define{interior multiplication} by $u$ is the grade -1 map: 
\[
    \iota_u\maps \Omega^p(\M) \to \Omega^{p-1}(\M) 
\]
defined on 1-forms $X$ by $\iota_u X=X(u)$, and on higher forms by requiring it to be a graded derivation:
\ben
\iota_u(X\wedge Y)=(\iota_u X)\wedge Y+(-)^p X\wedge \iota_u Y\,,
\een
where $X$ is a $p$-form. 

\begin{defn} 
\label{def:temp}
We say a differential form $X$ on $\M$ is \define{temporal} if $\hat u \we X = 0$, and \define{spatial} if $\iota_u X = 0$. 
\end{defn}
It is then easy to check that any form $X$ splits into spatial and temporal parts as:
\ben
  X = \underbrace{(X - \hat u \we \iota_u X) }_\text{\sf spatial} +\underbrace{(\hat u \we \iota_u X) }_\text{\sf temporal} 
           =: X^\perp +  X^\parallel 
\een
and that the spatial and temporal projections are grade 0 derivations.   Similarly, we can define spatial and temporal differentials that act on differential forms as
\ben
d^{\perp}X=dX-\hat{u}\wedge\pounds_u X\,,\quad d^{\parallel}X=\hat{u}\wedge \pounds_u X\,,
\label{different}
\een
where $\pounds_u = \iota_u d + d \,\iota_u$ is the usual Lie derivative.

So far, all of this assumes a metric given from the outset.  As we have seen, we can avoid this by starting with a `fake frame bundle' $\ff$ and its associated fake tangent bundle $\fake$ with metric $\eta$.
\begin{defn} 
Given a fake frame bundle $\ff$ with associated `fake observer space' $\fo \cong P\times_H H/K$, a \define{field of internal observers} is a section of the bundle $\fo \to M$.   
\end{defn}
\noindent Such an internal observer field $y$ reduces $\ff$ to a principal $K$ bundle: $P$ is a principal $K$ bundle over $\fo$, and this can be pulled back along $y$ to a principal $K$ bundle $Q_y\to \M$: 
\ben
\label{reduction}
\xymatrix{
Q_y \ar[r]^{i_y} \ar[d] & \ff \ar[d] \\
\M \ar[r]^{y} & \fo
}
\een
Thinking of $\ff$ as a fake frame bundle, $Q_y$ corresponds to the bundle of frames that map a given (fixed) observer $\hk_o\in\ggh$ into $y(x)\in\fake_x$, with an obvious inclusion map $i_y$ into $\ff$.

The internal observer field also lets us split fields living in any associated vector bundle of $P$ into various components, according to how the relevant $H$ representation splits when pulled back to $K$.  For instance, the fake tangent bundle $\fake$ splits into {\em internal}  `temporal' and `spatial' parts: 
\ben
     \fake \cong P \times_{H} \ggh \cong (Q_y \times_{K}\ggh_o) \oplus_\M (Q_y \times_{K} \vec \ggh)\,,
\een
where $\oplus_\M$ denotes the fiberwise direct sum of vector bundles.   Similarly, the internal observer splits fields valued in the bundle $\Ad(P) = P \times_H \h$ into a $\k$ part and a $\ghk$ part, according to (\ref{split2}). 

Next, suppose we have not only a fake frame bundle, but a Cartan connection on it.  The $\ggh$ part gives us a coframe field, hence a specific isomorphism $T\M \cong \fake$.   Using this isomorphism, a field of internal observers obviously corresponds to a field of observers for the metric and orientation induced by the coframe field.  This same linking of `internal' and `spacetime' observers is the key to our construction of covariant Ashtekar variables \cite{lorentz}, which we also review in section \ref{cartangeo} with the benefit of observer space.   But first, we consider what happens to the Cartan connection when we reduce the frame bundle to a $K$ bundle via an observer field. 

For a given observer field $y\maps\M\to\fo$, we can use the inclusion $i_y$ in (\ref{reduction}) to pull back a Cartan connection $A$ on $\ff\to\M$ to $Q_y$:
\ben
\tilde{A}=i_y^*A:TQ_y\rightarrow\g\,.
\een
It is clear that $\tilde{A}$ cannot be a Cartan connection for the bundle $Q_y\to\M$, since $T_q Q_y$ and $\g$ do not have the same dimension. But, it has all other essential features of a Cartan connection:
\begin{prop}
\label{prop:pullb}
Let $(\pi\maps \ff \to \M, A)$ be a spacetime Cartan geometry with model $(G,H)$, and let $i\maps Q \to P$ be a reduction of $P$ to a principal $K$-bundle.  The pullback $\tilde{A}:=i^*A$ of the Cartan connection $A$ satisfies:
\begin{enumerate}
\item For each $q\in Q$, the projection of $\tilde{A}$ to $\k\oplus\ggh$, $(\tilde{A}_q)_{\k\oplus\ggh}\maps T_q Q\to \k\oplus\ggh$ is a linear isomorphism;
\item $R_k^\ast \tilde{A} = \Ad(k^{-1})\tilde{A} \quad \forall k \in K$;
\item $\tilde{A}$ restricts to the Maurer--Cartan form on vertical vectors.  
\end{enumerate}
\end{prop}

\begin{proof}
Note that the inclusion $i$ is $K$-equivariant: $i\circ R_k=R_k\circ i$. It then follows from $(R_k)^\ast A = \Ad(k^{-1}) A$ that $(R_k)^* \tilde A = (R_k)^* i^* A = i^*  (R_k)^* A = i^* (\Ad(k^{-1})A) = \Ad(k^{-1})\tilde A$, which shows 2). Similarly $i_*$ maps vertical vectors in $Q \to\M$ to vertical vectors in $\ff\to\M$, which shows 3) for $\tilde{A}$ since it holds for $A$. Lastly, to show the first property it suffices to show that $(\tilde{A}_q)_{\k\oplus\ggh}$ is injective. If $(\tilde{A}_q)_{\k\oplus\ggh}(V_q)=0$, then $A_{i(q)}(i_* (V_q))\in\ghk$. But the preimage of $\ghk$ under $A$ is the space of vectors canonically associated with $X\in\ghk$. For $i_*(V_q)$ to be in this space we must have $V_q=0$.
\end{proof}

This proposition has an immediate corollary that works only for the {\em Minkowski} observer space model, with $G\cong \ISO(3,1)$.  In this case $\k\oplus\ggh$ is a Lie subalgebra---the Lie algebra of the subgroup $J=\SO(3)\ltimes \R^4$ consisting of rotations and spacetime translations. The proposition therefore implies $(Q \to\M,\;\tilde{A}_{\k\oplus\ggh})$ is a Cartan geometry with model $(J,K)$.  For the other models, there is no subgroup analogous to $J$: according to  (\ref{commutators}), the bracket of a spatial translation in $\vec{\ggh}$ and a time translation in $\ggh_o$ is a boost, so that any subgroup including all spacetime `translations' must also include boosts.  Thus, for the models with $\Lambda\neq 0$, we do not get a Cartan geometry of this sort.  

On the other hand, using only spatial translations, and discarding the time translations, we can get a Cartan connection for any of the three models---not on spacetime but on {\em space}.  This leads to the subject of `Cartan geometrodynamics'.

\subsection{Cartan geometrodynamics and covariant Ashtekar variables}
\label{cartangeo}

Assume we have a Cartan connection on $\ff\to\M$ and an internal observer field $y$.  This gives us the observer field $u$, as well as the principal $K$ bundle $Q_y \to M$.  Also assume that there is a totally spatial hypersurface $\S$ (``space'') of $\M$, meaning a codimension 1 submanifold such that each $T_x\S$ consists entirely of spatial vectors: $\hat{u}(v)=0$ for any $v\in T_x\S$. Such $\S$ exists whenever the Frobenius condition $\hat{u}\wedge d\hat{u}=0$ holds. 

Pulling back the fake frame bundle $\ff$ along the embedding of $\S$ into $\M$ defines a spatial fake frame bundle $\ff_\S$ over $\S$, a principal $K$ bundle. Because of the way we defined $\S$, this is the bundle of fake frames 
\ben
 \ff_\S \cong\{  \text{proper linear isometries }  E\maps \R^3 \to (\fake_x)^{y},\; x\in \S  \}\,.
\een
where $(\fake_x)^{y}$ is the subspace of $\fake_x$ orthogonal to $y(x)$ in the metric $\eta$ on $\fake$. 

We have an inclusion map $i_{y,\S}\maps\ff_\S\to Q_y$ which maps $E$ to its unique extension to an isometry $\ggh \to \fake_x$ respecting time orientation, resulting in the following diagram:
\ben
\xymatrix{
\ff_\S \ar[r]^{i_{y,\S}} \ar[d] & Q_y \ar[r]^{i_y} \ar[d] & \ff \ar[d] \\
\S \ar[r] & \M \ar[r]^{y} & \fo
}
\een
Pulling back the 1-form $\tilde{A}=i_y^*A$ on $Q_y$ along $i_{y,\S}$ gives a $\g$-valued 1-form on $\ff_\S$ which now always defines a Cartan geometry.
\begin{prop}
\label{prop:pullb2}
The projection ${\bf A}$ of the 1-form $i_{y,\S}^*\tilde{A}=i_\S^* i_y^* A$, the pullback of the 1-form $\tilde{A}$ on $Q_y\to M$ to the bundle $\ff_\S\to\S$, to $\k\oplus\vec\ggh$ is a Cartan connection; $(\ff_\S\to\S,\;{\bf A})$ is a Cartan geometry modeled on the Klein geometry $(H',K)$ where $H'$ is the Lie group with Lie algebra $\k\oplus\vec\ggh$.

\end{prop}

\begin{proof}
Again we need to show that ${\bf A}$ satisfies properties \ref{cartan:nondeg}-\ref{cartan:mc} in definition \ref{def:cartan}. As in proposition \ref{prop:pullb}, the second property follows from compatibility of the inclusion with the action of $K$, and the third from the fact that the inclusion maps vertical vectors to vertical vectors. For the first property, by proposition \ref{prop:pullb}, the projection of the image of $(i_{y,\S})_*T_r\ff_\S$ (for $r\in\ff_\S$) under $\tilde{A}$ to $\k\oplus\ggh$ is a six-dimensional subspace of $\k\oplus\ggh$. It contains all of $\k$ by property 3, and a three-dimensional spacelike subspace of $\ggh$ in the metric $\eta$. The projection of such a subspace to $\vec\ggh$ must also be three-dimensional. 
\end{proof}

The Cartan connection ${\bf A}$ on $\ff_\S\to\S$ is the basic ingredient for the picture of Ashtekar variables as Cartan geometrodynamics. If we have not just one spatial hypersurface $\S$ but a foliation of $\M$ by totally spatial hypersurfaces of identical topology, we can, just as in usual geometrodynamics, identify the hypersurfaces and view them as spatial geometries evolving in time, where `time' is the function $t$ associated to the field of co-observers $\hat{u}$ which is of the form $\hat{u}=N\,dt$ if the Frobenius condition $\hat{u}\wedge d\hat{u}=0$ holds. On each of these `constant time slices' $\S_t$ we have a Cartan geometry given by $(\ff_{\S_t}\to\S_t,\;{\bf A}_t)$ modeled on $H'/K$; we call this \define{Cartan geometrodynamics}. 

In our previous discussion of Cartan geometrodynamics \cite{lorentz} the Cartan connection on each spatial slice was modeled on hyperbolic space $\Hyp^3$ since, without the backdrop of observer space, this was the only 3-dimensional observer space in sight.  Here, we clearly see the role of the group $H'$, which depends on the cosmological constant chosen for the model spacetime $\GH\cong G/H$.  The Cartan geometry on each spatial slice is most naturally modeled on hyperbolic space, Euclidean space or the sphere, according to whether the spacetime model is anti-de Sitter, Minkowski, or de Sitter, respectively.   

We now detail the construction of Lorentz-covariant Ashtekar variables \cite{lorentz} in order to clarify the relation between Ashtekar variables and Cartan geometrodynamics. The $\g$-valued connection $A$ on the $K$-bundle $P\to\fo$ breaks up into the irreducible pieces as in (\ref{irredpie}):
\[
\begin{array}{ccccccccc}
 \g & = & \k &\oplus & \ghk & \oplus & \vec\ggh & \oplus & \ggh_o \\
 A & = & \Omega & + & b & + & \vec e & + & e_o
\end{array}
\]
Pulling back $A$ along the internal observer field $y$, we obtain the connection $\tilde{A}$ on $Q_y\to\M$ which splits similarly:
\ben
\begin{array}{ccccccccc}
\tilde{A} & = & y^* \Omega & + & y^* b & + & E & + & \hat{u}\,.
\end{array}
\label{decompo1}
\een

The parts of $\tilde{A}$ valued in $\vec\ggh$ and $\ggh_o$ play a special role: they determine the notions of spatial and temporal vectors and forms.   In particular, spatial vectors live in the kernel of $\hat{u}$, while temporal vectors live in the kernel of the \define{triad} $E$; notions of spatial and temporal forms follow from this.   We can define the field of observers $u$ by requiring it to be spatial and normalized by $\hat u$:
\ben
E(u)=0\,,\quad \hat{u}(u)=1\,.
\een
All this agrees with definition \ref{def:temp}.

We can then split $y^*\Omega$ and $y^* b$ into their spatial and temporal parts,
\ben
{\bf\Xi}:= y^*{\Omega}(u)\,,\quad {\bf \Omega} := y^*\Omega - \hat{u}\,{\bf\Xi}\,,\quad \xi:=y^*b(u)\,,\quad K := y^*b - \hat{u}\,\xi\,,
\label{decompo2}
\een
to obtain the variables needed to describe generalized canonical gravity in \cite{lorentz} (where we only defined the $\h$-valued scalar $\Xi={\bf\Xi}+\xi$). Note that ${\bf\Omega}(u)=K(u)=0$, as required.

Starting from the Palatini action for general relativity as a functional of the $\h$ and $\ggh$ parts $\omega$ and $e$ of a Cartan connection (identified with $\tilde{A}$) and decomposing the variables and their derivatives further according to (\ref{decompo1}), (\ref{decompo2}), and (\ref{different}), we then recover the usual formulation of canonical gravity in connection variables.  The dynamical variable conjugate to the triad (or rather the inverse triad $[E,E]$) is the $\k$-valued \define{Ashtekar-Barbero connection}
\ben
A_{{\rm A.B.}} := {\bf \Omega} + \gamma \star K
\een
where $\star$ is the Hodge dual on $\h$ and $\gamma$ is the Immirzi parameter, giving the relative weight between the two terms in the inner product used to define the Palatini action; cf. (\ref{forminner}) and the discussion below. The variables $E,\bf\Omega$ and $K$ must satisfy the Gauss constraints
\ben
[K,[E,E]]=0\,,\quad d^{\perp}_{\bf\Omega}[E,E]=0\,,
\label{gauss}
\een
where the spatial covariant differential $d^{\perp}_{\bf\Omega}$ acts as $d^{\perp}_{\bf\Omega}X=d^{\perp}X+[{\bf\Omega},X]$. They are also subject to the diffeomorphism and Hamiltonian constraints which can be given in the form
\ben
\left[E,\mathfrak{R}[{\bf\Omega}]+\half[K,K]+d^{\perp}\hat{u}\,{\bf\Xi}\right]=0\,,\quad \left[E,d^{\perp}_{\bf\Omega}K + d^{\perp}\hat{u}\,\xi\right]=0
\een
where $\mathfrak{R}[{\bf\Omega}]:=d^{\perp}{\bf\Omega}+{\bf\Omega}\wedge{\bf\Omega}$ is the spatial curvature of ${\bf\Omega}$. If the temporal 1-form $\hat{u}$ defines a foliation of spacetime, $d^{\perp}\hat{u}=0$ and the terms involving $d^{\perp}\hat{u}$ disappear; then one recovers the constraint formulation of Ashtekar-Barbero variables. See \cite{lorentz} for more details.

By proposition \ref{prop:pullb2}, if $\hat{u}$ defines a foliation, the fields ${\bf\Omega}$ and $E$ can be assembled into an $\h'$-valued Cartan connection ${\bf A}$ that defines {\em spatial} geometry as a Cartan geometry modeled on $(H',K)$. Using the Ashtekar-Barbero connection instead of ${\bf\Omega}$ defines a different spatial Cartan geometry of the same type with Cartan connection ${\bf A}_{{\rm A.B.}} = (A_{{\rm A.B.}},E)$.

In \cite{lorentz} we argued that Cartan geometrodynamics gives an understanding of the Ashtekar-Barbero formulation of canonical gravity as a theory of spontaneously broken $H$-symmetry. From the perspective of observer space advocated in the present paper, it would be more appropriate to speak of a spontaneously broken $G$-symmetry, where group $G$ is spontaneously broken to $K$ by a choice of internal observer field.

\sect{General relativity on observer space}
\label{sec:lifting-gr} 

\subsection{Reconstructing spacetime}
\label{reconstructing}

We have seen in lemma \ref{lem:cart} that spacetime Cartan geometry automatically gives us Cartan geometry on observer space.  We now deal with the more interesting converse question: Given just an observer space $O$, under what conditions can we sensibly construct a spacetime $\M$ for which $O$ is the unit future tangent bundle?

To address this question, first note that in the principal $H$ bundle $\ff\to \M$, the spacetime manifold $\M$ itself actually contains only redundant information: it is just the space of $H$-orbits $\ff/H$.  In the observer space picture, however, we do not, a priori, have any action of $H$ on $\ff$.  Rather, we have only a $K$ action, since $\ff$ is a principal $K$ bundle over observer space.  However, while there is no action of the group $H$, we will see that under certain conditions, we get an action of the {\em Lie algebra} $\h$.  Just as an action of a Lie group $H$ on a manifold $M$ is a group homomorphism 
\[
   H \to   \Diff(M)\,,
\]
an \define{action of a Lie algebra} $\h$ on a manifold $M$ is a Lie algebra homomorphism
\[
   \h \to   \Vect(M)
\]
where $\Vect(M)$ is the Lie algebra of $\Diff(M)$, the Lie algebra of smooth vector fields on $M$.   Any Lie algebra action $\alpha\maps \h \to \Vect(M)$ is integrable in the sense that the distribution $\alpha(\h)$ is an integrable distribution.  The integral submanifolds of this distribution are the \define{orbits} of the $\h$-action, and the space of orbits is denoted $M/\h$.  \cite{alek-mic}

In our situation, we have an action of $K$ on the manifold $\ff$, and this induces an action of the Lie algebra $\k$ of $K$, by differentiation.  Using the Cartan connection $A$, it is clear that the distribution $\underline A(\k)$ on $\ff$ is just the distribution tangent to the fibers of $\ff \to \ff/K$.  So, the spaces of orbits coincide:
\[
        \ff/\k = \ff/K\,. 
\]  
To reconstruct spacetime, we can try to extend this $\k$-action to an $\h$-action.  Whenever this works, we can immediately define:
\[
    \text{\define{spacetime} } \M := \ff/\h. 
\]
In the best case, the $\h$-action will integrate to an $H$-action, the map $\ff \to \M$ will become a principal $H$ bundle, and the observer space Cartan connection will therefore give the geometry of spacetime.  

To see how Cartan connections are related to Lie group actions, it is helpful to first view Cartan connections in a different way, using the nondegeneracy of a Cartan connection $A\maps TP \to \g$ to turn this map around.  More precisely, suppose that $A$ is any $\g$-valued 1-form on a manifold $P$ such that $A_p\maps T_pP \to \g$ is a linear isomorphism (i.e., property \ref{cartan:nondeg} in the definition of a Cartan connection).  Then we get a map
\[
\uA\maps \g \to \Vect(P)
\]
where for $X \in \g$, the value of $\underline A(X)$ at $p \in P$ is 
\begin{equation}
\label{inverse-conn}
\uA(X)_p:= (A_p)^{-1}(X).
\end{equation}
It is easy to check that the properties \ref{cartan:nondeg}, \ref{cartan:equivariant}, and \ref{cartan:mc} are respectively equivalent to:
\begin{C'-list}
\item \label{cartan:nondeg2} For each $p\in P$, the composite $\g\stackto{\uA} \Vect(P) \to T_pP$ is a linear isomorphism;

\item \label{cartan:equivariant2}    $R_{h*} \circ \underline A = \underline A\circ \Ad(h^{-1}) \quad \forall h \in H$;
\item \label{cartan:mc2}  For $X$ in $\h$, $\underline A(X)$ is the canonical vertical vector field associated to $\h$.  
\end{C'-list}
We recall that any element of $\h$, thought of as a left-invariant vector field on $H$, pulls back canonically to a vector field on $P$ along any local trivialization.  

If $A$ is a Cartan connection, $\uA$ is generally {\em not} a Lie algebra homomorphism.  In fact, the failure of $\uA$ to be a homomorphism corresponds to the {\em curvature} $F$ of $A$: 
\begin{prop}
\label{homo-fail}
Suppose $({\pi\maps P\to \M},A)$ is a Cartan geometry modeled on the Klein geometry $(G,H)$ and define $\underline A$ by {\rm (\ref{inverse-conn})}.   Then
\[
 \uA([X,Y]) - [\uA(X),\uA(Y)] =  \uA\big(F(\uA(X),\uA(Y))\big). 
\]
for all $X,Y\in \g$.
\end{prop}

\begin{proof}
For any 1-form $\omega$, $d\omega(X,Y) = X(\omega(Y)) - Y(\omega(X)) - \omega([X,Y])$, where $X(\omega(Y))$ denotes the directional derivative of the function $\omega(Y)$ along $X$.  Applying this formula to $dA(\uA(X),\uA(Y))$, the first two terms are directional derivatives of constant functions, and thus vanish, leaving simply:
\[
    dA(\uA(X),\uA(Y)) 
    = - A([\uA(X),\uA(Y)]).
\]  
Hence, with $F= dA + \half[A,A]$, 
\begin{align*}
   F(\uA(X),\uA(Y)) &= dA(\uA(X),\uA(Y)) + [X,Y] \\
   &= - A([\uA(X),\uA(Y)]) + [X,Y]
\end{align*}
After applying $\uA$ to both sides, we have the result. 
\end{proof}

This proposition implies that $\uA$ restricts to a homomorphism on the subalgebra $\h\subseteq \g$.  In fact, we get a bit more: 
\ben
\label{inf-conn}
     [\uA(X),\uA(Y)] = \uA([X,Y])
\een
whenever at least {\em one} of $X,Y$ is in $\h\subseteq \g$.  To see this note that curvature $F$ is horizontal: it vanishes on any vertical vector.  Thus, if {\em either} $X$ or $Y$ is in $\h$, then the left hand side of the equation in the lemma vanishes.  In particular, $\uA|_\h\maps \h \to \Vect(P)$ is an action of the Lie algebra $\h$ on $P$.  These observations lead to a more general notion of Cartan connection studied by Alekseevsky and Michor \cite{alek-mic2} where the principal bundle structure is discarded, defining a \define{$\g/\h$-Cartan connection} to be a 1-form on a manifold $P$ for which (\ref{inf-conn}) holds whenever at least one argument is in $\h$. A significant part of Cartan's theory carries over to this case.   

In our case, we start with a Cartan geometry for the model $G/K$ and want a Cartan geometry for the model $G/H$. 
\begin{defn}
\label{def:h-flat}
Suppose $({\pi\maps P\to O},A)$ is a Cartan geometry modeled on the Klein geometry $(G,K)$, and $\h$ is a Lie algebra with $\k\subseteq \h \subseteq \g$.  We say $A$ is \define{$\h$-flat} if
\ben
     [\uA(X),\uA(Y)] = \uA([X,Y])
\label{h-flat}
\een
whenever at least {\em one} of $X,Y$ is in $\h\subseteq \g$.  We say $A$ is \define{$\h$-complete} if $\uA(X)$ is a complete vector field (i.e.\ generates a global flow) for all $X\in \h$. 
\end{defn}
\noindent Note that a Cartan geometry modeled on $(G,K)$ is trivially $\k$-flat, and $\g$-flat if and only if it is flat. 

We can now give conditions under which spacetime geometry can be reconstructed from observer space.  To state these conditions it is most convenient to modify the model geometry slightly, by replacing the groups $G$, $H$ and $K$ with their universal covering groups.  In each of the cases we are interested in, the universal cover is just the double cover:
\[
  \Gtild = \left\{ 
  \begin{array}{rl}
  \Spin_o(4,1)  \\
  \ISpin_o(3,1)  \\
  \Spin_o(3,2) 
  \end{array}
  \right. \qquad \quad
  \Htild = \Spin_o(3,1) 
   \qquad \quad
  \Ktild = \Spin(3) 
\]
where $\ISpin(3,1) := \Spin(3,1)\ltimes \R^{3,1}$, and the subscript ${}_o$ denotes the connected component, as before.  Note that passing to these covering groups does not change the geometry at all on the infinitesimal level, since the Lie algebras are unchanged.  

\begin{thm}[{Reconstruction of spacetime}]
\label{thm:ex-spacet}
Suppose $({\pi\maps P\to O},A)$ is an observer space geometry, with model $(\Gtild,\Ktild)$. 
\begin{enumerate}
\item If $A$ is $\h$-flat, then $\uA(\h)$ spans an integrable distribution of constant rank.
\item If in addition $A$ is $\h$-complete, it induces a locally free action of $\Htild$ on $P$.  
\item If this locally free action is free and proper, then the quotient $M=P/\Htild$ is a smooth manifold, and $({\pi\maps P\to M},A)$ is a Cartan geometry with model $(\Gtild,\Htild)$.  
\end{enumerate}
\end{thm}  
\noindent In the third part, recall that a proper $G$ action on $X$ is one for which the graph $G \times X \to X \times X$ defined by $(g,x)\mapsto (x,gx)$ is a proper map: the preimages of compact sets are compact.
\vskip .8em
\begin{proof}
Since $A$ is a Cartan connection for a geometry on observer space, property \ref{cartan:nondeg} implies $\uA(\h)$ has constant rank equal to the dimension of $\h$.   If $A$ is $\h$-flat, then in particular $A|_\h$ is a Lie algebra action, and the span of any Lie algebra action is integrable.    

If $A$ is also $\h$-complete, then a result of Palais \cite{palais} implies the Lie algebra action $A|_\h$ is the derivative of a locally free action of some Lie group $\overline H$ with Lie algebra $\h$.  But the action of $\overline H$ induces an action of its universal cover $\Htild$.  

If $\Htild$ has a free and proper action on $P$, then there exists a unique smooth structure on $P/\Htild$ such that $P \to P/\Htild$ is a submersion; with this smooth structure, $P\to P/\Htild$ is a principal $\Htild$-bundle.  (see e.g.~\cite[Thm.~1.21]{meinrenken}).  It remains to check the properties of a Cartan connection.  The nondegeneracy property \ref{cartan:nondeg} is already satisfied.  Since $\Htild$ is connected, equivariance \ref{cartan:equivariant} follows from $\h$-flatness.  For property \ref{cartan:mc}, note that restricted to any fiber of $P\to P/\Htild$, $A$ is a complete $\h$-valued 1-form that is an isomorphism on each tangent space and satisfies the Maurer--Cartan equation.  Hence, the fiber is an $\Htild$ torsor and $A$ restricts to the Maurer--Cartan form.  (See e.g.\ \cite[Thm 8.7]{Sharpe} for uniqueness of the Maurer--Cartan form with respect to these properties.)   
\end{proof}

It will be convenient to rephrase $\h$-flatness for an observer space Cartan connection $A$ in terms of  $A$ itself, rather than $\uA$: 
\begin{prop}
Suppose $({\pi\maps P\to O},A)$ is a Cartan geometry on observer space, modeled on the Klein geometry $(G,K)$.  Then $A$ is $\h$-flat if and only if $F$ vanishes on any $v\in TP$ for which $A(v)\in \ghk$. 
\end{prop}
\begin{proof}
Since $\uA$ is nondegenerate at each point, it is clear from Prop.~\ref{homo-fail} that $A$ is $\h$-flat if and only if $F(\uA(X),\uA(Y)) = 0$ whenever $X \in \h$.  Since $A$ is a Cartan connection modeled on $(G,K)$ this equation holds already for $X \in \k$, and hence we require only that 
\[
   F(\uA(X),\uA(Y)) = 0 \qquad \forall X\in \ghk\,.
\]
But $\uA(\g)$ spans the tangent space at each point of $P$, and $\uA(\ghk)$ spans the subspace of `boost' vectors $v$.  So, the condition on $F$ is equivalently written 
\[
   F(v,w) = 0 \qquad \forall v \in TP \text{ such that } A(v) \in \ghk\,,
\]
as we wished to show.
\end{proof}

To rephrase this proposition, $A$ is $\h$-flat if and only if $F(\ughk,v) = 0$ for all vectors $v$, where $\ughk = \uA(\ghk)$ is the distribution on $P$ corresponding to Lorentz boosts, discussed in section~\ref{observer-cartan}.  While $F$ need not vanish on $\ughk$, it is interesting to note that the $\ghk$ part of $F$ always does:
\begin{prop}
Suppose $({\pi\maps P\to O},A)$ is an observer space geometry. Then the $\ghk$ part of the curvature, $F_\ghk$, is a spacetime 2-form, i.e.\ it vanishes on any vector in the boost distribution $\ughk$.
\end{prop}
\begin{proof}
From (\ref{curvature-split}) we have $F_\ghk = d_\Om b +[\vec e, e_o]$.   We must show that $F_\ghk(v,w) = 0$ whenever $A(v)\in \ghk$.  It suffices to pick $\underline{v}\in\ghk$ and $\underline{w}\in\g$ and define $v=\uA(\underline{v})$, $w=\uA(\underline{w})$.  In fact, it suffices to check the case where $\underline{w}\in \ghk \oplus \ggh$, since we know that $F(v,w)$ vanishes whenever $\underline{w}\in \k$, since $A$ is a Cartan connection with model $(G,K)$.  We have:
\begin{align*}
F_\ghk (v,w) &= (db + [\Om,b] + [\vec e,e_o])(v,w) \\
&= v(b(w)) - w(b(v)) + [\Om(v),b(w)] - [\Om(w),b(v)] + [\vec e(v),e_o(w)] - [\vec e(w),e_o(v)]\,.
\end{align*}
The first two terms are directional derivatives of constant functions, and so vanish.  Moreover, since $v\in \uA(\ghk)$ and $w \in \uA(\ghk \oplus \ggh)$, we have $\Om(v) = \Om(w) = \vec e (v) = e_o(v)=0$.
\end{proof}

\subsection{Action}

We have seen in section \ref{sec:macdo} how to define general relativity in terms of the MacDowell-Mansouri action (\ref{macdow}) for a Cartan connection $A$ on spacetime $\M$, corresponding to a Cartan geometry modeled on $(G,H)$. Here we give a new action on observer space $O$, starting with a Cartan connection $A$ corresponding to a Cartan geometry modeled on $(G,K)$. In order to use theorem \ref{thm:ex-spacet}, our action enforces $\h$-flatness of $A$; the rest of the action is a straightforward extension of the MacDowell-Mansouri action (\ref{macdow}). Using a new field to enforce $\h$-flatness, we find that this field in general appears as a source to Einstein's equations, so that only a certain class of solutions to our action will correspond to vacuum general relativity.

We define the following action on observer space:
\bena
S[A,\lambda,\chi]&=&\int_O  \kappa_{\h} \left(F_\h\wedge F_\h\right)\wedge \tau_{\ghk}(b\wedge b\wedge b)\nonumber
\\&& + \left[\tr_{\g}\left(F(\lambda)\right)+\tr_{\h\otimes\g}\left(\chi\,[e, e](\lambda)\right)\right]\tr_{\h}([e,e]\wedge\star[e,e])\wedge \tau_{\ghk}(b\wedge b\wedge b)\,.
\label{graction}
\eena
Here, $\kappa_{\h}$ is a general $H$-invariant inner product on $\h$ (as defined in (\ref{forminner})), by `$\tr$' we again mean the Killing form, and $A$ is the Cartan connection on observer space, which splits according to the splitting of representations (\ref{irredpie}) which we repeat here: 
\[
\begin{array}{ccccccccc}
 \g & = & \k &\oplus & \ghk & \oplus & \vec\ggh & \oplus & \ggh_o \\
 A & = & \Omega & + & b & + & \vec e & + & e_o
\end{array}
\]
In particular, $b$ denotes the $\ghk$ part and $e$ is the part valued in $\ggh= \vec\ggh \oplus \ggh_o$.  The other two fields include a bivector valued in $\g$
\ben
    \lambda\in \Lambda^2 TO\tensor \g
\een
and a scalar valued in $\h \tensor \g$:
\ben
    \chi \maps O \to \h\tensor \g\,. 
\een 
Since $F$ is a $\g$-valued 2-form, we thus have $F(\lambda)\in \g \tensor \g$, and the bilinear form $\tr_\g\maps \g \tensor \g\to \R$ gives us a scalar.  Similarly, since $e\in T^* O\tensor \ggh$, the commutation relation $[\ggh,\ggh]\subseteq \h$ gives us $[e,e] \in \Omega^2(O) \tensor \h$.  Feeding in the bivector $\lambda$, we thus get $[e,e](\lambda)\in \h\tensor \g$, and hence a scalar after applying $\tr_{\h\tensor\g}\maps \h\tensor\g\tensor\h\tensor\g \to \R$ to $\chi\,[e,e](\lambda)$. 
We also need to fix a $K$-invariant trilinear form $\tau_{\ghk}$ on $\ghk$, which as a representation of $K$ is isomorphic to the adjoint representation of $K$; in component notation, one may take this to be $\tau_{\ghk}\left(b\wedge b\wedge b\right):= \epsilon_{ijk}b^i\wedge b^j\wedge b^k$. As one may have expected, the action (\ref{graction}) is only invariant under $K$, not $H$.

First consider the equations of motion resulting from variation with respect to $\chi$ and $\lambda$:
\bena
[e, e](\lambda) & = & 0\,,\nonumber
\\F & = & - \tr_{\h}\left(\chi \,[e,e]\right)\,.
\label{equofmo}
\eena
To interpret these equations, note that since $A$ is a Cartan connection corresponding to a Cartan geometry modeled on $(G,K)$, its projection on $\ghk\oplus\ggh$ defines a `siebenbein', an isomorphism between each tangent space to observer space and $\ghk\oplus\ggh$. In particular, $e$ defines a basis of `spacetime' 1-forms at each point in observer space. Then $\tr_{\h}([e,e]\wedge\star[e,e])$ is just the usual `spacetime' volume form induced by the vierbein $e$, and the wedge product with $\tau_{\ghk}(b\wedge b\wedge b)$ defines a volume form on observer space.

The second equation in (\ref{equofmo}) then says that the curvature $F$ vanishes on any `boost' vector, i.e.\ on any vector in the 3-dimensional subspace annihilated by the span of $e$. This is precisely the subspace of vectors that are mapped into $\ghk$ by the Cartan connection $A$, and gives us the condition of $\h$-flatness (\ref{h-flat}). The remaining components of $F$ are left arbitrary; they are given in terms of the Lagrange multiplier $\chi$. The first equation in (\ref{equofmo}) is a restriction on $\lambda$; it requires its `spacetime' components to vanish.

On the solutions of (\ref{equofmo}), we therefore have a spacetime $\M$ whose observer space is $O$, and $A$ is a Cartan connection defining a Cartan geometry modeled on $(G,H)$.

To obtain the equations of motion satisfied by the Cartan connection $A$, one can split the variations of (\ref{graction}) with respect to the connection $A$ into variations with respect to its `spacetime' form part and the rest. Variation with respect to the `spacetime' form part will give, schematically,
\bena
\delta S & = & \int_O  \delta\Big(\kappa_{\h} \left(F_\h\wedge F_\h\right)\Big)\wedge \tau_{\ghk}(b\wedge b\wedge b)\nonumber
\\&& + \delta\Big(\left[\tr_{\g}\left(F(\lambda)\right)+\tr_{\h\otimes\g}\left(\chi\,[e, e](\lambda)\right)\right]\Big)\tr_{\h}([e,e]\wedge\star[e,e])\wedge \tau_{\ghk}(b\wedge b\wedge b)\nonumber
\\&& + \left[\tr_{\g}\left(F(\lambda)\right)+\tr_{\h\otimes\g}\left(\chi\,[e, e](\lambda)\right)\right]\delta\Big(\tr_{\h}([e,e]\wedge\star[e,e])\wedge \tau_{\ghk}(b\wedge b\wedge b)\Big)\,.
\eena
The first line reproduces the equations of motion of MacDowell-Mansouri gravity, wedged with the everywhere non-zero 3-form $\tau_{\ghk}(b\wedge b\wedge b)$. The contribution from the third line vanishes once (\ref{equofmo}) is imposed, since then the terms in square brackets are zero. But the variation of these terms need not be zero on-shell; the second line in general gives a contribution to the equations of motion involving the covariant divergence in `velocity space' directions of the mixed `spacetime'/`velocity space' components of $\lambda$, which thus appears as a source in Einstein's equations.  Here we are facing well-known issues with trying to enforce nonholonomic constraints through Lagrange multipliers; see e.g.\ \cite{nonholo} for a general discussion. In order to restrict to vacuum general relativity (with cosmological constant), we must assume this divergence to vanish everywhere. Physically, we could impose the requirement that the field $\lambda$ on observer space is really a `spacetime' field, parallelly transported along `velocity' directions by the connection $A$. We did not however find an elegant way to impose this condition directly through the action.

The remaining variation is with respect to the `velocity' form parts of $A$ which gives
\bena
\delta S & = &\int_O \kappa_{\h} \left(F_\h\wedge F_\h\right)\wedge \delta\Big(\tau_{\ghk}(b\wedge b\wedge b)\Big) \nonumber
\\ && + \delta\Big(\tr_{\g}\left(F(\lambda)\right)\Big)\tr_{\h}([e,e]\wedge\star[e,e])\wedge \tau_{\ghk}(b\wedge b\wedge b)\,;
\eena
if we restrict to solutions where the source term involving $\lambda$ vanishes, the 4-form $\tr_{\h} \left(F_\h\wedge {\star F_\h}\right)$ vanishes, so that there is no contribution from the first term and we only get further restrictions on the field $\lambda$.

\sect{Relative spacetime} 
\label{sec:beyond}

Up to now, while the definition of an observer space geometry does not presuppose the existence of spacetime, we have been largely concerned with recovering spacetime from observer space.   In section \ref{reconstructing} we showed an $\h$-flat Cartan connection on observer space gives a concept of  spacetime that all observers agree on: spacetime is the quotient of observer space obtained by integrating out the `boost' distribution on observer space.  But this is a rather special situation: the boost distribution of an arbitrary observer space geometry is generically nonintegrable.  We now consider the more general situation.  When the boost distribution is not integrable, spacetime itself is at best an observer-dependent approximation.  Briefly, spacetime is relative.  

In considering the idea of relative spacetime, we are leaving the conceptual foundations of general relativity, and we make no attempt at specific prescriptions for how the general relativistic observer space should be deformed.  Rather, we emphasize generic features of observer space, which reduce to familiar notions in the case of an integrable boost distribution, but are just as sensible in general.  

In a general observer space geometry, each observer has a set of observers perceived as being at the same `spacetime point'.
\begin{defn}
\label{def:here-now}
Let $(\pi\maps P \to O, A)$ be an observer space Cartan geometry.  The set of observers \define{coincident} to a given observer $o\in O$ is the union of points along geodesics starting at $o$ with initial velocity $v$ such that $A(v)\in\ghk$. 
\end{defn}
\noindent
The notion of `coincidence' thus becomes relative, analogously with the relativity of simultaneity in special relativity: For a nonintegrable distribution of `boost' vector fields, two different observers might view the same third observer as coincident without viewing each other as coincident.  For observer spaces in which the deviation from $\h$-flatness is appropriately small, such effects would be noticeable only at very high relative velocities. 

Likewise, even when no absolute spacetime exists, each observer can reconstruct a local notion of `spacetime', not as a quotient, but as a subspace of observer space, much as `space' in special relativity is a particular observer-dependent slice through spacetime. 
\begin{defn}
The \define{local spacetime} according to a given observer $o\in O$ is the union of points in observer space along geodesics starting at $o$ with initial velocity $v$ such that $A(v)\in\ggh$. 
\end{defn}
\noindent
We emphasize that this is a local definition: far away from the initial observer both topological and geometric problems may arise.  Notice that for an observer space constructed from a spacetime, such geodesics project down to geodesics in spacetime; these geodesics may not fill all of spacetime, but they fill at least some neighborhood of the spacetime point corresponding to the observer.

We now turn to examples of observer spaces with no underlying notion of spacetime.  As we have emphasized, this is the generic situation, and only special observer spaces have a notion of spacetime.  However it is worth considering examples that are special in other ways.  

Let us consider observer spaces with no absolute spacetime, but with an absolute notion of {\em velocity space}.  Such examples are motivated by the recently proposed `principle of relative locality' \cite{relative}.   Much like in our description of observer space, in this proposal phase space is deformed in such a way that spacetime is no longer a natural quotient but rather an observer-dependent subspace.  This `phase space'---presumably the symplectification of some observer space not associated to any spacetime---has so far mostly been studied mostly under the simplifying assumption that there is an {\em absolute momentum space}.  Underlying this momentum space is an {\em absolute velocity space}, the space of momenta with unit mass.

To construct such examples, we choose a three-dimensional Riemannian manifold $V$ to serve as velocity space.  To fit experimental constraints, its geometry should deviate only slightly from that of hyperbolic space.  We then proceed similarly to the construction in section \ref{sec:observerspace} of observer spaces from spacetime.  We use the Minkowski observer space model, $\ISO_o(3,1)/\SO(3)$, since this is the only one with an absolute velocity space.  In particular, the velocity space of the model is
\ben
\label{ineff} 
\ISO_o(3,1)/J \cong\SO_o(3,1)/\SO(3),
\een
where $J\cong\SO(3)\ltimes\R^4$ is the group of rotations for a fixed observer and spacetime translations.  

Using the model spacetime as a guide, we arrive at a canonical procedure for constructing an observer space from our velocity space $V$.  Recall that an observer space with absolute spacetime is a sub-bundle of the tangent bundle of spacetime, with three-dimensional fiber at each spacetime point.  Likewise, an observer space with absolute velocity space will be an extension of the tangent bundle of velocity space, with four-dimensional fiber at each point.  In the Minkowski model, this just means we can view the observer space $\R^{3,1}\times \Hyp^3$ as a bundle over either spacetime $\R^{3,1}$ or velocity space $\Hyp^3$. 

More precisely, our Riemannian 3-manifold $V$ has a canonical hyperbolic Cartan geometry, with model $H/K \cong \SO_o(3,1)/\SO(3)$, built from the Levi--Civita connection and coframe, just as in section \ref{lorentz-to-cartan}.  The Cartan connection is an $\h$-valued 1-form on the principal $K$ bundle $FV$ of orthonormal frames.  Using the isomorphism (\ref{ineff}), we extend this canonically to a 1-form $A$ on the principal $J$ bundle $\ff=FV\times_K J$ over $V$, with values in the Lie algebra $\g\cong \Iso(3,1)$.  So far, this is only a rather redundant reformulation of the Cartan geometry of $V$: it uses the model of $\Hyp^3$ as an $\ISO_o(3,1)$-space, where the translation part acts trivially, rather than the more effective model as an $\SO_o(3,1)$-space. 

However, we now have what we want.  The associated bundle $$O:=\ff \times_J \R^{4}\cong \ff \times_J J/K$$ is an affine extension of the tangent bundle of $V$, the projection $\ff \to O$ is a principal $K$ bundle, and the Cartan connection $A$ gives an observer space geometry on $O$ (cf. lemma \ref{lem:cart}).   This observer space is $\j$-flat by construction and has absolute velocity space $P/\j\cong V$.  However, `spacetime' for a particular observer is the fiber in $O$ over that observer's velocity;  there is in general no canonical way to identify these notions of spacetime.  

This procedure suggests a way to construct further examples: whenever a particular model observer space has some particular absolute feature, we can consider Cartan-geometric deformations which preserve that feature.   For example, in the de Sitter models, there is an absolute notion of `conformal infinity'.  This is given by an $\SO_o(4,1)$-equivariant map from the de Sitter observer space to the conformal 3-sphere, a homogeneous $\SO_o(4,1)$-space with parabolic stabilizer.  Deformations maintaining an abolute notion of conformal infinity are the ones whose Cartan connections remain flat in the directions determined by the Lie algebra of this parabolic subgroup.  In this way, while maintaining the coherence of conformal infinity, we arrive at observer spaces that have no notion of absolute spacetime, nor of absolute velocity space.  This is the subject of forthcoming work \cite{holo}.

\section{Conclusions and outlook}

As indicated in the introduction, one of our motivations has been to relate covariant and canonical approaches to gravity.  
The discord between these two pictures has led some physicists, beginning with Dirac (see \cite{kragh}, p.\ 290) to doubt the ultimate significance of the spacetime picture:
\begin{quote}
This result has led me to doubt how fundamental the four-dimensional requirement in physics is.  \ldots 
[I]t seems that four-dimensional symmetry is not of overriding importance, since the description of nature sometimes gets simpler when one departs of it.  
\end{quote}
  This doubt has perhaps been carried furthest by Barbour, whose work has culminated in an alternative to general relativity in which only space is fundamental and spacetime emerges from the theory itself \cite{shape}.  Other theories under current investigation, including the anisotropic gravity of Ho\v{r}ava \cite{horava}, and causal dynamical triangulations \cite{cdt} start from a spacetime picture, but introduce a preferred spatial foliation, restoring an absolute notion of simultaneity. 

In this paper, we have argued for a complementary approach: rather than assuming any fundamental concept of `space', we take the notion of an {\em observer} seriously, as ontologically prior to either space or spacetime.   As observers, we do not experience spacetime directly.  From our collective experience, we notice that:
\begin{itemize}
\item Each of us can organize the objects near us by describing their positions using three coordinates.  In brief, each observer sees `space' as three-dimensional.  
\item We each experience things changing in time.   Each observer sees `time' as one-dimensional.   
\item We can relate other observers to us according to their (relative) velocity.  We see `velocity space'---the space of all observers coincident with us---as three-dimensional.  
\end{itemize}
Admittedly, these do not appear to be independent, since we use our notions of space and time to measure velocities.  However, measuring a velocity requires nonlocal measurements, and even special relativity shows that the `obvious' relationship among space, time, and velocity is only approximate.  
Most important is relativity's lesson that we {\em disagree} on these notions of space, time, and velocity space.  If we suppose that an observer is uniquely determined by its perceived notions of these three spaces, then we need, a priori, some seven-dimensional space of observers.   It is not obvious that this can be reduced precisely to some 4-dimensional `spacetime'.  

Observer space provides a new way of understanding the geometry of general relativity.  We have also argued that the flexibility of this new perspective provides a natural setting for studying proposed modifications of general relativity, and their relationships.  This leaves much to be done, both on the subject of observer space itself, and on applications. 

First, it will be interesting to study more particular examples of observer space, both those arising from solutions of general relativity and observer spaces without an underlying spacetime. For instance, one immediate question is how the existence of an event horizon in the spacetime of a black hole, or of a cosmological horizon in an expanding universe, is encoded in the geometry of observer space. 

Second, while lightlike particles play an obviously important role in general relativity, we have so far mostly ignored the extension of observer space that includes them, as described at the end of section \ref{spacetime-observer-spaces}.   The action of the group $H$ on hyperbolic space $\Hyp^3 \cong H/K$ can be extended to an action on the compactification $\overline{\Hyp^3}$, with two orbits: $\Hyp^3$ itself and the boundary.  The respective stabilizers  are $K$ and $K'\cong \SIM(2)$, the stabilizer of a light ray through the origin of Minkowski spacetime.  A general {\em extended} observer space geometry should include not only Cartan geometry modeled on $G/K$, as we have defined in this paper, but also Cartan geometry modeled on $G/K'$, describing the boundary of observer space.   In the same way that the $G/K$ geometry is related to standard Hamiltonian methods, as described in section \ref{cartangeo}, the $G/K'$ geometry should presumably be related to light-front methods \cite{lightfront}, in which the splitting of fields is done from the perspective of `lightlike observers'.   A deeper study of this geometry may reveal connections to other theories, such as the `very special relativity' proposal, which uses $\SIM(2)$ as the fundamental spacetime symmetry group \cite{vsr}. 

To discuss possibly observable Lorentz violation, it will be useful to couple matter fields to gravity on observer space. For gravity itself, as we have mentioned in the introduction, there exist several competing ideas that are not covariant under changes of observer. The perspective of observer space could allow studying those ideas from a new geometric angle. 

On the phenomenological side, it would be interesting to tighten the relationship between observer space and the relative locality proposal.  While the two seem clearly related, and both lead to the idea that locality, or coincidence, is a relative notion, the two frameworks have different starting points.  In particular, the idea behind relative locality involves building up spacetime geometry from the {\em interactions} of particles in the universal velocity (or rather momentum) space.  

The idea that spacetime geometry is `velocity-dependent' or `momentum-dependent' appears in several approaches going beyond usual Lorentzian geometry as the basic framework for gravitational physics, some of which might be related to observer space Cartan geometries. The most obvious example is Finsler geometry, sometimes referred to as `Riemannian geometry without the quadratic restriction' \cite{chernfins}---one replaces the metric by general {\em length functional}, whose second derivative with respect to velocities can be viewed as a `velocity-dependent metric'. Observer spaces naturally describe a `velocity-dependent' geometry, although it is not completely obvious how to relate our connection-based approach to an essentially metric-based approach.

Relating our framework to possible predictions for physical measurements will involve clarifying some interpretational and foundational issues. Measurements should be made by inertial observers moving on certain geodesics on observer space (as discussed in section \ref{spacetime-observer-spaces}) and compared to those made by other observers. We then have to understand why the assumption that spacetime exists is compatible with our experience to such excellent precision. We saw that it is not obvious to recover an underlying spacetime from observer space when trying to give an action for general relativity on observer space. Similarly, we must explain why matter fields are not arbitrary functions on observer space, but to a good approximation just fields on `spacetime'. We leave all of this to future work.

\section*{Acknowledgments}
SG thanks the Institute for Quantum Gravity of the University of Erlangen-N\"urnberg for supporting a visit during which some of this work was carried out.  Likewise, DW thanks the Perimeter Institute for Theoretical Physics for supporting a visit, and especially Jim Dolan and Josh Willis for valuable discussions that helped lead up to our study of observer space.  Research at Perimeter Institute is supported by the Government of Canada through Industry Canada and by the Province of Ontario through the Ministry of Research \& Innovation.

\appendix
\section{Notation}
\label{notation}

Here are the letters used in this paper for various things.
\[
\begin{array}{cll}
\M &\text{spacetime} \\
O &\text{observer space} \\
x &\text{point in }\M  \\
o &\text{point in }O \\
\fake & \text{fake tangent bundle} \\
\ff & \text{fake orthonormal frame bundle} \\
\fo & \text{fake observer space} \\
\GH & \text{model spacetime}  \\
\gh & \text{point in }\GH \\
\ggh & T_\gh \GH \\
\end{array}
\]

\begin{flushright}preprint pi-mathphys-298\end{flushright}


\begin{thebibliography}{99}

\bibitem{alek-mic} D.\ V.\ Alekseevsky and P.\ W.\ Michor, Differential geometry of $\g$-manifolds, {\sl Differential Geom.\ Appl.\ }{\bf 5} (1995) 371--403, \arxiv{math/9309214}.

\bibitem{alek-mic2} D.\ V.\ Alekseevsky and P.\ W.\ Michor, Differential geometry of Cartan connections, {\sl Publ.\ Math.\ Debrecen} {\bf 47} (1995) 349--375, \arxiv{math/9412232}. 

\bibitem{cdt} J.\ Ambjorn, J.\ Jurkiewicz, and R.\ Loll, Reconstructing the universe, {\sl Phys.\ Rev.\ D\ }{\bf 72} (2005) 064014, \hepth{0505154}.

\bibitem{relative} G.\ Amelino-Camelia, L.\ Freidel, J.\ Kowalski-Glikman, and L.\ Smolin, The principle of relative locality, {\sl Phys.\ Rev.\ D} {\bf 84} (2011) 084010, \arxiv{1101.0931}.

\bibitem{arnold} V.\ I.\ Arnold, {\sl Mathematical Methods of Classical Mechanics}, 2ed. (Springer, New York, 1989). 

\bibitem{tele} J.\ C.\ Baez and D.\ K.\ Wise, Teleparallel gravity as a higher gauge theory, \arxiv{1204.4339}. 

\bibitem{barbero} J.\ F.\  Barbero G., 
Real Ashtekar variables for Lorentzian signature space times, 
{\sl Phys. Rev. D\ }{\bf 51}, 5507--5510 (1995). \grqc{9410014}.


\bibitem{shape} J. Barbour, Shape dynamics, an introduction, in {\sl Quantum Field Theory and Gravity} (Springer, Basel, 2012), \arxiv{1105.0183}; H.\ Gomes, S.\ Gryb, and T.\ Koslowski, Einstein gravity as a 3D conformally invariant theory, 	{\sl Class.\ Quant.\ Grav.\ }{\bf 28}, 045005 (2011), \arxiv{1010.2481}.

\bibitem{chernfins} S.\ S.\ Chern, Finsler geometry is just Riemannian geometry without the quadratic restriction, {\sl Notices AMS} {\bf 43} (1996) 959--63.

\bibitem{vsr} A.\ G.\ Cohen and S.\ L.\ Glashow, Very special relativity, {\sl Phys.\ Rev.\ Lett.} {\bf 97} (2006) 021601.

\bibitem{dida} H.\ M.\ Dida and A.\ Ikemakhen, A class of metrics on tangent bundles of pseudo-Riemannian manifolds, {\sl Arch.\ Math.\ (Brno)\ }{\bf 47} (2011) 293--308. 

\bibitem{nonholo} M.\ R.\ Flannery, The enigma of nonholonomic constraints, {\sl Am.\ J.\ Phys.\ }{\bf 73} (2005) 265-272.

\bibitem{dgr} G.\ W.\ Gibbons and S.\ Gielen, Deformed general relativity and torsion, {\sl Class.\ Quant.\ Grav.} {\bf 26} (2009) 135005, \arxiv{0902.2001}.

\bibitem{lorentz} S.\ Gielen and D.\ K.\ Wise, Spontaneously broken Lorentz symmetry for Hamiltonian gravity, {\sl Phys.\ Rev.\ D} {\bf 85} (2012) 104013, \arxiv{1111.7195}.

\bibitem{essay} S.\ Gielen and D.\ K.\ Wise, Linking covariant and canonical general relativity via local observers, {\sl Gen.\ Relativ.\ Grav.\ }{\bf 44} (2012) 3103-3109, \arxiv{1206.0658}.

\bibitem{lightfront} A.\ Harindranath, An introduction to light-front dynamics for pedestrians, in {\em Light-Front Quantization and Non-Perturbative QCD}, edited by J. P. Vary and F. Woelz (International Institute of Theoretical and Applied Physics, Ames, 1997),  \arxiv{hep-ph/9612244}. 

\bibitem{helgason} S.\ Helgason, {\sl Differential geometry, Lie groups, and symmetric spaces} (Academic Press, New York, 1978).

\bibitem{horava} P. Ho\v{r}ava, Quantum gravity at a Lifshitz point, {\sl Phys.\ Rev.\ D\ }{\bf 79} (2009) 084008, \arxiv{0901.3775}.

\bibitem{kibble} T.\ W.\ B.\ Kibble, Lorentz invariance and the gravitational field, {\sl J.\ Math.\ Phys.} {\bf 2} (1961) 212-221.

\bibitem{kragh} H.\ Kragh, {\sl Dirac: A Scientific Biography} (Cambridge University Press, Cambridge, 2005).

\bibitem{macdo} S.\ W.\ MacDowell and F.\ Mansouri, Unified geometric theory of gravity and supergravity, {\sl Phys.\ Rev.\ Lett.} {\bf 38} (1977) 739-742.

\bibitem{meinrenken} E.\ Meinrenken, Group actions on manifolds (2003), lecture notes available at \href{http://www.math.toronto.edu/mein/teaching/action.pdf}{http://www.math.toronto.edu/mein/teaching/action.pdf}.

\bibitem{michorbook} P.\ W.\ Michor, {\sl Topics in Differential Geometry (Graduate Studies in Mathematics Vol. 93)} (American Mathematical Society, 2008).

\bibitem{palais} R.\ S.\ Palais, {\sl A global formulation of the Lie theory of transformation groups}, Memoirs of the American Mathematical Society, Vol. 22 (American Mathematical Society, 1957).

\bibitem{Sharpe} R.\ W.\ Sharpe, {\sl Differential Geometry: Cartan's Generalization of Klein's Erlangen Program} (Springer, 1997).

\bibitem{stellewest} K.\ S.\ Stelle and P.\ C.\ West, Spontaneously broken de Sitter symmetry and the gravitational holonomy group, {\sl Phys.\ Rev.\ D} {\bf 21} (1980) 1466-1488.

\bibitem{derekmacd} D.\ K.\ Wise, MacDowell-Mansouri gravity and Cartan geometry, {\sl Class.\ Quant.\ Grav.} {\bf 27} (2010) 155010, \grqc{0611154}.

\bibitem{dereksigma} D.\ K.\ Wise, Symmetric space Cartan connections and gravity in three and four dimensions, {\sl SIGMA} {\bf 5} (2009) 080, \arxiv{0904.1738}.

\bibitem{broken} D.\ K.\ Wise, The geometric role of symmetry breaking in gravity, {\sl J.\ Phys.:\ Conf.\ Ser.\ }{\bf 360} (2012) 012017. \arxiv{1112.2390}.

\bibitem{holo} D.\ K.\ Wise, Holographic special relativity, in preparation.  

\end{thebibliography}
\end{document}